\documentclass[11pt]{article}

\def\final{1}  

\usepackage{xcolor}
\definecolor{ForestGreen}{rgb}{0.1333,0.5451,0.1333}
\definecolor{DarkRed}{rgb}{0.80,0,0}
\definecolor{Red}{rgb}{1,0,0}
\usepackage[linktocpage=true,
pagebackref=true,colorlinks,
linkcolor=DarkRed,citecolor=ForestGreen,
bookmarks,bookmarksopen,bookmarksnumbered]
{hyperref}

\usepackage{fullpage}
\usepackage{geometry}
\usepackage[T1]{fontenc}
\usepackage[utf8]{inputenc}
\usepackage[american]{babel}
\usepackage[normalem]{ulem}
\usepackage{amsmath, amssymb, cases, amsthm}
\usepackage{thmtools}
\usepackage[shortlabels]{enumitem}
\usepackage{mdframed}
\usepackage{bbm}
\usepackage{bm}
\usepackage{microtype}
\usepackage{xcolor}
\usepackage{makecell}
\usepackage{mathtools}
\usepackage{float}
\usepackage{multirow}
\usepackage{soul}

\usepackage{easyReview} 
\usepackage{graphics}
\usepackage{mathrsfs}

\usepackage[capitalize,noabbrev]{cleveref}

\usepackage{algorithmicx}
\usepackage[noend]{algpseudocode}
\algdef{SE}[SUBALG]{Indent}{EndIndent}{}{\algorithmicend\ }%
\algtext*{Indent}
\algtext*{EndIndent}

\usepackage[linesnumbered,ruled,vlined]{algorithm2e}
\SetKwInput{KwInput}{Input}                
\SetKwInput{KwOutput}{Output}              

\Crefname{claim}{Claim}{Claims}
\Crefname{condition}{Condition}{Conditions}
\Crefname{assumption}{Assumption}{Assumptions}
\Crefname{scenario}{Scenario}{Scenarios}

\usepackage{caption}
\usepackage{subcaption}

\usepackage{tcolorbox}

\declaretheorem[numberwithin=section,refname={Theorem,Theorems},Refname={Theorem,Theorems}]{theorem}
\declaretheorem[numberlike=theorem,name=Theorem,refname={Theorem,Theorems},Refname={Theorem,Theorems}]{thm}
\declaretheorem[numberlike=theorem]{lemma}
\declaretheorem[numberlike=theorem,name=Lemma,refname={Lemma,Lemmas},Refname={Lemma,Lemmas}]{lem}

\declaretheorem[numberlike=theorem,name=Proposition,refname={Proposition,Propositions},Refname={Proposition,Propositions}]{prop}
\declaretheorem[numberlike=theorem]{corollary}
\declaretheorem[numberlike=theorem,name=Corollary,refname={Corollary,Corollaries},Refname={Corollary,Corollaries}]{cor}
\declaretheorem[numberlike=theorem,style=definition]{definition}
\declaretheorem[numberlike=theorem,style=definition,name=Definition,refname={Definition,Definitions},Refname={Definition,Definitions}]{defn}
\declaretheorem[numberlike=theorem,style=definition]{invariant}
\declaretheorem[numberlike=theorem,style=definition]{condition}

\declaretheorem[numberlike=theorem,style=definition]{assumption}
\declaretheorem[numberlike=theorem]{claim}
\declaretheorem[numberlike=theorem,style=remark]{remark}

\declaretheorem[numberlike=theorem,refname={Fact,Facts},Refname={Fact,Facts},name={Fact}]{fact}

\declaretheorem[numberlike=theorem, refname={Observation,Observations},Refname={Observation,Observations},name={Observation}]{observation}

\newcommand{\vol}{\mathrm{vol}}
\renewcommand{\deg}{\mathrm{deg}}

\newcommand{\poly}{\mathrm{poly}}

\newcommand{\eps}{\varepsilon}

\newcommand{\rev}[1]{\overleftarrow{#1}}

\newcommand{\dist}{\mathrm{dist}}

\newenvironment{wrapper}[1]
{
	\begin{center}
		\begin{minipage}{\linewidth}
			\begin{mdframed}[hidealllines=true, backgroundcolor=gray!20, leftmargin=0cm,innerleftmargin=0.4cm,innerrightmargin=0.4cm,innertopmargin=0.4cm,innerbottommargin=0.4cm,roundcorner=10pt]
				#1}
			{\end{mdframed}
		\end{minipage}
	\end{center}
}

\global\long\def\bydef{\stackrel{\mathrm{def}}{=}}
\global\long\def\hi{\texttt{hi}}
\global\long\def\med{\texttt{med}}
\global\long\def\lo{\texttt{lo}}

\global\long\def\del{\mathrm{del}}

\global\long\def\dmg{\texttt{dmg}}
\global\long\def\base{\texttt{base}}
\global\long\def\safe{\texttt{sf}}
\global\long\def\most{\texttt{most}}
\global\long\def\del{\mathrm{del}}

\newcommand{\card}[1]{|#1|}

\newcommand{\Augment}{\textrm{Augment}}
\newcommand{\augment}{\textrm{Augment}}
\newcommand{\gres}{G_{\textrm{res}}}
\newcommand{\Gres}{\gres}
\newcommand{\tres}{T_{\mathrm{res}}}
\newcommand{\Tres}{\tres}
\newcommand{\es}{\textrm{\textsc{ES}}}
\newcommand{\ES}{\es}
\newcommand{\delete}{\textrm{Delete}}
\newcommand{\Delete}{\delete}
\newcommand{\Insert}{\mathrm{Insert}}

\newcommand{\Remove}{\mathrm{Remove}}

\newcommand{\distmax}{\dist_{\textrm{max}}}
\newcommand{\degmax}{\deg_{\textrm{max}}}

\newcommand{\LPM}{\textrm{LPM}}
\newcommand{\lpm}{\LPM}
\newcommand{\NumDeletions}{\textrm{\textsc{NumDeletions}}}
\newcommand{\epochsize}{q_{\textrm{ep}}}

\newcommand{\ResetES}{\textrm{\textsc{ResetES}}}
\newcommand{\Raff}{R_{\textrm{aff}}}

\newcommand{\Rmark}{R_{\textrm{mark}}}
\newcommand{\Lfar}{L_{\textrm{far}}}
\newcommand{\Vfar}{V_{\textrm{far}}}

\newcommand{\otil}{\tilde{O}}

\newcommand{\core}{\texttt{core}}
\newcommand{\res}{\texttt{adj}}
\newcommand{\hilo}{\texttt{hilo}}

\newcommand{\F}{\mathcal{F}}
\newcommand{\old}{\texttt{old}}
\newcommand{\new}{\texttt{new}}
\newcommand{\init}{\texttt{init}}
\newcommand{\vhi}{\texttt{v-hi}}
\newcommand{\alo}{\texttt{a-lo}}
\newcommand{\fnl}{\texttt{final}}

\global\long\def\Eu{\mathrm{Eu}}

\ifnum\final=0

\newcommand{\peter}[1]{{\color{purple}[\textbf{Peter}: #1]}}
\newcommand{\sayan}[1]{{\color{red}[\textbf{Sayan}: #1]}}
\newcommand{\aaron}[1]{{\color{blue}[\textbf{Aaron}: #1]}}

\def\thatchaphol#1{\marginpar{$\leftarrow$\fbox{TS}}\footnote{$\Rightarrow$~{\sf\textcolor{purple}{#1 --Thatchaphol}}}}

\else

\newcommand{\peter}[1]{}
\newcommand{\sayan}[1]{}
\newcommand{\aaron}[1]{}

\def\thatchaphol#1{}

\fi

\makeatletter
\newcommand\footnoteref[1]{\protected@xdef\@thefnmark{\ref{#1}}\@footnotemark}
\makeatother

\title{Deterministic Dynamic Maximal Matching\\ in Sublinear Update Time}

\author{Aaron Bernstein\thanks{
        New York University,
        \texttt{bernstei@gmail.com}. Supported by Sloan Fellowship, Google Research Fellowship,  NSF Grant 1942010, and Charles S. Baylis endowment at NYU.
    } \and Sayan Bhattacharya\thanks{
        University of Warwick,
        \texttt{S.Bhattacharya@warwick.ac.uk}.
    }\and Peter Kiss\thanks{
        University of Vienna,
        \texttt{peter.kiss@univie.ac.at} This research was funded in whole or in part by the Austrian Science Fund (FWF) 10.55776/ESP6088024, Part of this work was done while at the University of Warwick.
    }  \and Thatchaphol Saranurak\thanks{
        University of Michigan,
        \texttt{thsa@umich.edu}.
        Supported by NSF Grant CCF-2238138. 
        Part of this work was done while at INSAIT, Sofia University ``St. Kliment Ohridski'', Bulgaria. This work was partially funded from the Ministry of Education and Science of Bulgaria (support for INSAIT, part of the Bulgarian National Roadmap for Research Infrastructure).
    } }

\begin{document}

\begin{titlepage}
  \maketitle \pagenumbering{gobble}
  \begin{abstract}
We give a fully dynamic deterministic algorithm for maintaining a maximal matching of an $n$-vertex graph in  $\tilde{O}(n^{8/9})$ amortized update time. 
This breaks the long-standing $\Omega(n)$-update-time barrier on dense graphs, achievable by trivially scanning all incident vertices of the updated edge, and affirmatively answers a major 
 open question repeatedly asked in the literature \cite{BaswanaGS15,BhattacharyaCHN18,dagstuhl22461}.
 We also present a faster randomized algorithm against an adaptive adversary with $\tilde{O}(n^{3/4})$ amortized update time. 

Our approach employs the \emph{edge degree constrained subgraph} (EDCS), a central object for optimizing approximation ratio, in a completely novel way; we instead use it for maintaining a matching that matches \emph{all} high degree vertices in sublinear update time so that it remains to handle low degree vertices rather straightforwardly. 
To optimize this approach, we employ tools never used in the dynamic matching literature prior to our work, including sublinear-time algorithms for matching high degree vertices, random walks on directed expanders, and the monotone Even-Shiloach tree for dynamic shortest paths.
\end{abstract}
		
  \setcounter{tocdepth}{3}
  \newpage
  \phantom{a} 
  
  \vspace{-1.4cm}
  \tableofcontents
  \newpage
\end{titlepage}
\newpage
\pagenumbering{arabic}

\newpage

\section{Introduction}

A \emph{matching} $M$ of a graph $G=(V,E)$ is a set of vertex-disjoint edges. 
We say that $M$ is \emph{maximal} if every free edge $(u,v) \in E\setminus M$ is not vertex-disjoint from $M$. 
That is, a maximal matching is precisely a set of edges that are disjoint yet intersect all other edges.

Although computing a maximal matching admits a simple linear-time greedy algorithm, it presents interesting challenges in various computational models and has been extensively studied since the 80s, including in the PRAM model \cite{karp1985fast,luby1985simple,israeli1986fast,israeli1986improved,awerbuch1989network,alon1986fast,ghaffari2021time}, distributed models \cite{hanckowiak1999faster,hanckowiak2001distributed,barenboim2016locality,fischer2020improved,balliu2021lower}, and, more recently, the MPC model \cite{lattanzi2011filtering,ghaffari2018improved,czumaj2018round,ghaffari2019sparsifying,AssadiBBMS19,behnezhad2019exponentially}. In this paper, we study this problem in the dynamic setting.

\paragraph{Dynamic Maximal Matching.}

In the \emph{fully dynamic maximal matching} problem, we must maintain a maximal matching $M$ of an $n$-vertex $m$-edge graph $G$ undergoing both edge insertions and deletions. The goal is to minimize the \emph{update time }for updating $M$, given an edge update. Since the 90s, Ivkovi\'{c} and Lloyd \cite{ivkovic1993fully} showed the first non-trivial algorithm that has $O((n+m)^{0.7072})$ amortized update time and is deterministic. 

A key milestone was made by Baswana, Gupta, and Sen \cite{baswana2011fully,baswana2018fully} who showed a near-optimal randomized algorithm with $O(\log n)$ amortized update time. This was further improved to an optimal $O(1)$ update time by Solomon \cite{Solomon16}. Different $\poly(\log n)$-update-time algorithms with additional properties were given in \cite{behnezhad2019fully,BernsteinFH21,ghaffari2024parallel}. Unfortunately, all these algorithms share a common drawback; they are all randomized and assume an \emph{oblivious adversary}.\footnote{Dynamic algorithms assume an \emph{oblivious adversary }if they assume that the whole update sequence is fixed from the beginning and is independent of the answers maintained by the algorithm. Without this assumption, we say that the algorithms work against an \emph{adaptive} adversary.} 

The state of the art of deterministic algorithms is much worse. No previous algorithms could even strictly beat the following trivial solution: When $(u,v)$ is inserted, match $u$ and $v$ if both $u$ and $v$ were unmatched. When $(u,v)$ is deleted, check if $u$ or $v$ can be matched by scanning through $O(n)$ neighbors of $u$ and $v$. This straightforwardly takes $O(n)$ update time. While there are algorithms with $\tilde{O}(\alpha)$ worst-case update time where $\alpha$ is the arboricity \cite{neiman2016simple,christiansen2022adaptive}, they still require $\Omega(n)$ update time in dense graphs. Overcoming the $\Omega(n)$ deterministic barrier remains a long-standing open problem.

\paragraph{Success in Dynamic Approximate Maximum Matching.}

The influential result of \cite{baswana2011fully,baswana2018fully} also implies a fully dynamic algorithm for maintaining $2$-approximate maximum matching. Since then, the community has shifted its focus to fully dynamic \emph{approximate maximum} matching in various approximation regimes, including the approximation factors of $(2+\epsilon)$ \cite{BhattacharyaHN16,BhattacharyaCH17,BhattacharyaHN17,ArarCCSW18,CharikarS18,BhattacharyaHI18,BhattacharyaK19,BernsteinFH21,Kiss22}, $(2-\epsilon)$ \cite{BhattacharyaHN16,behnezhad2020fully,roghani2022beating,behnezhad2023dynamic,BhattacharyaKSW23,AzarmehrBR24}, $(1.5+\epsilon)$ \cite{BernsteinS15,BernsteinS16,GrandoniSSU2022,Kiss22}, $(1.5-\epsilon)$ \cite{Behnezhad23}, and $(1+\epsilon)$ \cite{GuptaP13,AssadiBKL23,LiuOffline}. Numerous papers dove into specialized topics, including partially dynamic algorithms \cite{gupta2014maintaining,grandoni20191+,BernsteinGS20,bhattacharya2023dynamicLP,jambulapati2022regularized,AssadiBD22,blikstad2023incremental,chen2023entropy,brand2024almost}, dynamic approximation-preserving reductions from weighted to unweighted graphs \cite{GuptaP13,StubbsW17,BernsteinDL21,bernstein2024matching}, and dynamic rounding algorithms \cite{ArarCCSW18,Wajc20,BhattacharyaK21,BhattacharyaKSW24,dudeja2024note}.

Within the last decade, these 40+ papers have made the dynamic matching problem one of the most actively studied dynamic graph problems. A significant part of this body of work \cite{BhattacharyaHN16,BhattacharyaCH17,BhattacharyaHN17,BhattacharyaHI18,BhattacharyaK19,Kiss22,grandoni2022maintaining} successfully developed deterministic algorithms whose update times match those of their randomized counterparts.

\paragraph{The Barrier Remains.}

But despite the successes mentioned above, the $\Omega(n)$ deterministic barrier for dynamic maximal matching remains unbroken. A high-level explanation is that a maximal matching is far more fragile than an approximate matching. For approximate matching, via standard reductions \cite{AssadiKL19,Kiss22}, it requires $\Omega(\epsilon n)$ edge updates before the maximum matching size changes by a $(1+\epsilon)$ factor. In contrast, a single edge deletion can destroy maximality completely. 

This challenge has become a major open problem in the field and has been repeatedly posed as an open question \cite{BaswanaGS15,BhattacharyaCHN18,dagstuhl22461}. The problem remains unresolved even for randomized algorithms against an \emph{adaptive} adversary. 

In this paper, we break this long-standing $\Omega(n)$ deterministic barrier.
\begin{thm}
\label{thm:main}There is a \emph{deterministic} fully dynamic maximal matching algorithm with $\tilde{O}(n^{8/9})$ amortized update time. Also, there is a randomized fully dynamic maximal matching algorithm with $\tilde{O}(n^{3/4})$ amortized update time that works with high probability against an \emph{adaptive adversary}.\footnote{The $\tilde{O}$ notation hides $\poly\log n$ factors.}
\end{thm}

Thus, we give the first deterministic maximal matching algorithm with sublinear update time. Using randomization but \emph{without} an oblivious adversary, we can speed up the algorithm even further.

In general, significant effort has been made towards closing the gap between oblivious and adaptive adversaries in dynamic graph problems, which include dynamic connectivity \cite{holm2001poly,nanongkai2017dynamic,wulff2017fully,nanongkaiSW2017dynamic,goranci2021expander,chuzhoy2020deterministic}, sparsifiers \cite{bernstein2020fully,chen2020fast,goranci2021expander,bhattacharya2022simple,van2022faster,kyng2024dynamic,haeupler2024dynamic}, and shortest paths \cite{bernstein2016deterministic,bernstein2017deterministic,bernsteinC2017deterministic,chuzhoy2019new,gutenberg2020decremental,gutenberg2020deterministic,chuzhoy2021deterministic,chuzhoy2023new,kyng2024dynamic,haeupler2024dynamic}. These algorithms are crucial building blocks in the modern development of fast static graph algorithms \cite{chen2022maximum,abboud2022breaking,abboud2023all,brand2024almost}. We take the first step in this direction for dynamic maximal matching. 

\subsection*{Techniques}

\paragraph{Previous Techniques.}

The key question in designing dynamic maximal matching algorithms is how to handle deletions of a matched edge, as it is relatively easy to handle edge insertions or deletions of an unmatched edge. A trivial solution is that, whenever a vertex $v$ is unmatched, we scan through $v$'s neighbors and try to match $v$. This takes $O(\deg(v))=O(n)$ update time. 

The previous randomized algorithms \cite{baswana2018fully,solomon2016fully,behnezhad2019fully,ghaffari2024parallel} speed this up by trying to ensure that the matched edges are not deleted too often. For intuition, consider the following simplistic scenario when the graph is a star with root $u$ and leaves $v_{1},\dots,v_{n}$, and the adversary only deletes edges. Suppose we sample a random edge $(u,v_{i})$ and include it in the matching. If \emph{the adversary is oblivious to our random choices}, then it will take $\Omega(n)$ edge deletions in expectation before the adversary deletes our matched edge $(u,v_{i})$. Once this happens, even if we spend $O(n)$ time, we can charge this cost to the $\Omega(n)$ deletions, leading to $O(1)$ amortized update time. This is \emph{the} basic idea that all previous randomized algorithms exploited and successfully carried out beyond this simplistic scenario. Unfortunately, it completely fails without an oblivious adversary, and, in particular, for deterministic algorithms.

\paragraph{A Bird's-Eye View of Our Algorithm.}

Our deterministic algorithm is very different. We maintain a vertex set $V^{\star}$ in $o(n)$ update time such that
\begin{enumerate}
\item \label{enu:birdeye low deg}the maximum degree in $G[V\setminus V^{\star}]$ is $o(n)$\footnote{Technically, we actually show that $G[V\setminus V^{\star}]$ has arboricity $o(n)$.}, and 
\item \label{enu:birdeye perfect}we can further maintain a \emph{$V^{\star}$-perfect} matching $M_{\base}$ that matches \emph{all} vertices in $V^{\star}$ in $o(n)$ update time.
\end{enumerate}
Given this, we can maintain a maximal matching $M\supseteq M_{\base}$ in $o(n)$ update time. Indeed, if vertex $v$ becomes unmatched, we only need to scan through $v$'s neighbors in $V\setminus V^{\star}$, which has $o(n)$ vertices, as all vertices in $V^{\star}$ are already matched by $M_{\base}$. Observe that this argument crucially requires that \emph{all} vertices of $V^{\star}$ are matched. Maintaining a perfect matching is normally very difficult in the dynamic setting and there are strong conditional lower bound for this task \cite{HenzingerKNS15,dahlgaard2016hardness}. So we will need to pick a $V^{\star}$ that has additional structure.

\paragraph{New Applications of EDCS.}

Surprisingly, we can obtain $V^{\star}$ from an \emph{edge-degree constrained subgraph} (EDCS), a well-known subgraph $H\subseteq G$ that is useful for $(1.5+\epsilon)$-approximate maximum matching algorithms in the dynamic setting and beyond \cite{bernstein2015fully,bernstein2016faster,assadi2019towards,assadi2019coresets,grandoni2022maintaining,azarmehr2023robust,behnezhad2023dynamic,bhattacharya2022sublinear,Bernstein24}. See \Cref{defn:edcs} for definition.
The outline below presents a completely novel application of this central object in the literature. 

The set $V^{\star}$ is simply the set of vertices with ``high degree'' in the EDCS $H$, which can be explicitly maintained in $o(n)$ time using existing results \cite{GrandoniSSU2022}. By the structure of EDCS, \Cref{enu:birdeye low deg} follows directly. To see why \Cref{enu:birdeye perfect} should hold, we observe that the graph $H_{\hilo}:=E_{H}(V^{\star},V\setminus V^{\star})$ has a \emph{degree gap}. More precisely, for some number $X$ and $\gamma>0$, every vertex in $V^{\star}$ has degree at least $X$ in $H_{\hilo}$, while every vertex in $V\setminus V^{\star}$ has degree at most $(1-\gamma)X$ in $H_{\hilo}$. Consequently, for every set $S\subseteq V^{\star}$, we have $|N_{H_{\hilo}}(S)|\ge(1+\Omega(\gamma))|S|$, i.e., Hall's condition holds \emph{with a slack}. This strong expansion implies the existence of short $O(\log(n)/\gamma)$-length augmenting paths and allows us to maintain a $V^{\star}$-perfect matching in $o(n)$ update time. 

\paragraph{New Toolkit for Dynamic Matching.}

To carry out the above approach, we employ tools never used in the dynamic matching literature prior to our work. For example, to maintain the $V^{\star}$-perfect matching deterministically, we use the monotone Even-Shiloach tree for dynamic shortest paths \cite{EvenS81,King99,henzinger2016dynamic} and, for our faster randomized algorithm, we apply random walks on directed expanders. 
Although not necessary for breaking the $\Omega(n)$ barrier, we also apply a a sublinear-time algorithm for matching high-degree vertices \cite{assadi2024faster} to further optimize our update time.

\paragraph{Concurrent Work.}
Computing maximal matching is among the fundamental symmetry breaking problems, that also include computing maximal independent set, vertex coloring and edge coloring. 
Very recently, Behnezhad, Rajaraman and Wasim \cite{behnezhad2025dynamiccoloring} showed the first randomized fully dynamic algorithm for maintaining a $(\Delta+1)$-vertex coloring against an adaptive adversary in sub-linear update time. 

\paragraph{Organization.} 
We first present the preliminaries in \Cref{sec:notations} and the subroutine for efficiently matching most vertices with almost maximum degree in \Cref{sec:match most}. They are needed for our detailed overview in \Cref{sec:overview}.
Then, we present a complete algorithm for maintaining perfect matching of high-degree vertices in a graph with degree gap in \Cref{sec:decremental} and use it to formally prove our main result in \Cref{new:dyn:maximal}.

\section{Preliminaries}
\label{sec:notations}

\paragraph{Basic Notations.} Consider any graph $G = (V, E)$. We let  $N_G(v) \subseteq E$ denote the set of edges in $G$ that are incident on a node $v \in V$, and we define $\deg_G(v) := |N_G(v)|$ to be the degree of $v$ in $G$. 

\paragraph{Matching.}
A matching $M \subseteq E$ in a graph $G = (V, E)$ is a subset of edges that do not share any common endpoint. We let $V(M)$ denote the set of nodes matched under $M$. Furthermore, for any set $S \subseteq V$, we let $M(S) := \{ (u, v) \in M : \{u, v \} \cap S \neq \emptyset\}$ denote the set of matched edges under $M$ that are incident on some node in $S$.

Let $M$ be a matching in $G$. We say that a vertex $v$ is $M$-free if $v\notin V(M)$, otherwise it is $M$-matched. An \emph{augmenting path} $P$ in $G$ with respect to $M$ is a path in $G$ such that the endpoints are $M$-free and its edges alternate between edges in $M$ and edges not in $M$. Let $M \oplus P$ stand for $(M \setminus P) \cup (P \setminus M)$, the matching obtained by extending $M$ with the augmenting path $P$. We have $M\oplus P$ is a matching of size $|M|+1$.  Every $M$-matched vertex is also matched in $M\oplus P$.

\paragraph{EDCS.}
For any graph $H$ and $e=(u,v)$ where $e$ is not necessary an edge of $H$, the degree of $e$ in $H$ is defined as $\deg_{H}(e)\bydef\deg_{H}(u)+\deg_{H}(v)$. The basis of our algorithm is an edge degree constrained subgraph (EDCS). 
\begin{defn}
\label{defn:edcs}
    An $(B,B^{-})$-EDCS $H$ of $G$ is a subgraph of $G$ such that
\begin{itemize}
\item For each edge $e\in H$, $\deg_{H}(e)\le B$, and
\item For each edge $e\in G\setminus H$, $\deg_{H}(e)\ge B^{-}$.
\end{itemize}
\end{defn}

\begin{thm}
[\cite{GrandoniSSU2022}]\label{thm:edcs}For any parameters $B\le n$ and $\epsilon>0$, there is a deterministic algorithm that, given a graph $G$ with $n$ vertices undergoing both edge insertions and deletions, explicitly maintains a $(B,B(1-\epsilon))$-EDCS $H$ of $G$ using $O(\frac{n}{B\epsilon})$ worst-case update time. Furthermore, there are at most two vertices $u,v \in V$ such that their degree $\deg_H(v), \deg_H(u)$ in the EDCS changes due to the handling of an update.
\end{thm}


\paragraph{Graph Access.}
We say that an algorithm has \emph{adjacency matrix query} access to graph $G$ if it can answer queries with inputs $(u,v) \in V \times V$ returning true iff $(u,v) \in E$ in $O(1)$ time. We say that an algorithm has \emph{adjacency list query} access to graph $G$ if it can find the degree of any vertex and answer queries with input $(v,i) \in V \times [\deg_G(v)]$ returning the $i$-th neighbour of vertex $v$ in $G$ according to some arbitrary ordering in $O(\log n)$ time. 


\section{Match Most Vertices with Almost Maximum Degree}
\label{sec:match most}

\thatchaphol{General comment: we use ``it is easy to'' and ``it is not hard to'' quite often. I am afraid people will be annoyed.}
Next, we give a useful observation that there always exists a matching that matches most all vertices with almost maximum degree. Our deterministic and randomized implementations of the algorithm of Lemma~\ref{lem:match almost max degree} rely on a black box application of \cite{elkin2024deterministic} and a white box application of \cite{assadi2024faster}. As the lemma can mostly be concluded from existing results in literature we defer its proof to Appendix~\ref{app:proof:lem:match almost max degree}.
\begin{lem}
\label{lem:match almost max degree}Let $G=(V,E)$ be a graph with maximum degree at most $\Delta$ which can be accessed through both adjacency list and matrix queries. Let $V_{\kappa}$ denote the set of all nodes in $G$ with degree at least $(1-\kappa)\Delta$. Then, for any $\kappa > 0$ there exists a matching $M \subseteq E$ that matches all but $2\kappa n$ many nodes in $V_{\kappa}$. 
\begin{itemize}
\item If $\Delta\ge\frac{4}{\kappa}$, then $M$ can be computed in $O(n\Delta\kappa\log^2 n)$  time with high probability.
\item $M$ can be deterministically computed in $O(m\frac{\log n}{\kappa}) = O(n \Delta \cdot \frac{\log n}{\kappa})$ time. 
\end{itemize}
\end{lem}




\begin{remark}
While the specific bounds of Lemma \ref{lem:match almost max degree} rely on recent advances in static coloring algorithms, any bound of the form $O(m\poly(\Delta/\kappa))$ would translate to  a sublinear update time for our main result of fully dynamic maximal matching. For example, one could instead use the classical coloring algorithm of Gabow~et~al.~\cite{gabow1985edgecoloring}, or one could start with a trivial fractional matching and use existing tools to round it to an integral matching one (see e.g. \cite{BhattacharyaKSW24} ).
\end{remark}

\begin{remark}
The randomized algorithm of Lemma~\ref{lem:match almost max degree} can be improved to have a running time of $\tilde{O}(n)$ using a white-box modification of \cite{assadi2024faster}. As this running time difference doesn't ultimately end up affecting the update time of our final algorithm, we opted for the slower subroutine of Lemma~\ref{lem:match almost max degree}, because it can be obtained from \cite{assadi2024faster} in a black-box way.
\end{remark}

\section{Overview: Our Dynamic Maximal Matching Algorithm}
\label{sec:overview}
\label{sec:decremental:overview}

For ease of exposition, we first present the algorithm in a decremental setting, where the input graph only undergoes edge deletions. We will show how to handle fully dynamic updates without increasing update time at the end  in \cref{sec:fullydynamic:overview}.


\paragraph{The Framework.} 
We fix two parameters $B \in [n]$ and $\epsilon \in (0, 1)$ whose values will be determined later on, and define 
\begin{equation}
\label{overview:eq:delta}
\delta := 100 \epsilon.
\end{equation}
The reader should think of $\epsilon = 1/n^{\beta}$ and $B = 1/\epsilon^2$ for some absolute constant $0 < \beta < 1$.

\medskip
Let $G = (V, E)$ be the input  graph with $n$ nodes, undergoing edge deletions. We maintain a $(B, (1-\epsilon)B)$-EDCS $H := (V, E(H))$ of $G$ at all times, as per~\cref{thm:edcs}. This incurs an update time of $O\left(\frac{n}{B\epsilon}\right)$. The EDCS $H$ is maintained explicitly, and  we can make adjacency-list and adjacency-matrix queries to $H$ whenever we want. We maintain the adjacency-lists of all vertices in $H$ as binary search trees. Hence, we may answer an adjacency-list and adjacency-matrix query to $H$ in $O(\log n)$ and $O(1)$ time respectively.

Our dynamic algorithm  works in {\bf phases}, where each phase consists of $\delta n$ consecutive updates in $G$. Let $G_{\init} = (V, E_{\init})$ and $H_{\init} = (V, E(H_{\init}))$ respectively denote the status of the input graph $G$ and the EDCS $H$ at the start of a given phase. Throughout the duration of the phase, we let $H_{\core}$ denote the subgraph of $H_{\init}$ restricted to the edges in $G$. Specifically, we define $H_{\core} := (V, E(H_{\core}))$, where $E(H_{\core}) := E(H_{\init}) \cap E$. Note that $H_{\core}$ is a decremental graph within a phase, i.e., the only updates that can occur to $H_{\core}$ are edge deletions.

\paragraph{Organization of the Overview.} In~\cref{sec:overview:assume}, we present a classification of nodes based on their degrees in the EDCS $H$ at the start of a phase, under one simplifying assumption (see~\cref{assume:gap}). To highlight the main ideas, we first describe and analyze our algorithm under this assumption.~\cref{sec:overview:safe} shows how to maintain a \emph{base matching} $M_{\base}$, which matches \emph{all} the nodes of a certain type (that have large EDCS degrees). Next,~\cref{overview:sec:adjunct} shows how to maintain a maximal matching in the subgraph of $G$ induced by the nodes that are free under $M_{\base}$. In~\cref{sec:together}, we put together the algorithms from~\cref{sec:overview:safe,overview:sec:adjunct}, and derive~\cref{thm:main}. Next, in~\cref{sec:assume:gap}, we explain how to deal with the scenario when the simplifying assumption we made earlier (see~\cref{assume:gap}) does not hold.~\cref{sec:outline:proof} sketches the proof outline of a key lemma, which allows us to maintain the matching $M_{\base}$ efficiently in~\cref{sec:overview:safe}. Finally,~\cref{sec:fullydynamic:overview} shows how to extend our algorithm so that it can deal with fully dynamic updates.

\subsection{Classification of Nodes and a Simplifying Assumption}
\label{sec:overview:assume}

We will outline how to maintain a maximal matching $M_{\fnl}$ in $G$ during a phase. We start by making one simplifying assumption stated below. Towards the end of this section, we  explain how to adapt our algorithm when this assumption does not hold (see~\cref{sec:assume:gap}).

\begin{assumption}
\label{assume:gap} For every node $v \in V$, either $\deg_{H_{\init}}(v) > \left(\frac{1}{2} + \delta \right)B$ or $\deg_{H_{\init}}(v) < \left(\frac{1}{2} - \delta \right)B$.
\end{assumption}

In other words, this establishes a small ``degree-gap'' within the EDCS $H$ at the start of the phase, by asserting that there does \emph{not} exist any node $v \in V$ with $\left(\frac{1}{2} - \delta\right)B \leq \deg_{H_{\init}}(v) \leq \left(\frac{1}{2}+\delta \right)B$. We say that a node $v \in V$ is {\bf very high} if $\deg_{H_{\init}}(v) > \left(\frac{1}{2} + \delta \right)B$, and {\bf low} if $\deg_{H_{\init}}(v) < \left(\frac{1}{2} - \delta \right)B$. We let $V_{\vhi}$ and $V_{\lo}$ respectively denote the sets of very high and low nodes. By~\cref{assume:gap}, the node-set $V$ is partitioned into $V_{\vhi} \subseteq V$ and $V_{\lo} = V \setminus V_{\vhi}$. Our algorithm will crucially use two properties that arise out of this classification of nodes, as described below.

\begin{prop}
\label{overview:prop:fact}The following properties hold.
\begin{enumerate}
\item \label{overview:enu:very-high} Consider any edge $(u, v) \in E(H_{\init})$ where $u$ is very-high. Then $v$ must be low. 
\item \label{overview:enu:low} Consider any edge $e  \in E_{\init}$ whose both endpoints are low. Then  $e$ must appear in $H_{\init}$.
\end{enumerate}
\end{prop}

\begin{proof} First, we prove \cref{overview:enu:very-high}. Note that since $(u, v) \in E(H_{\init})$, we have $\deg_{H_{\init}}(u) + \deg_{H_{\init}}(v) \leq B$ (see~\cref{defn:edcs}). Furthermore, since $u$ is very-high, we have $\deg_{H_{\init}}(u) > \left( \frac{1}{2} + \delta \right)B$. Thus, it follows that $\deg_{H_{\init}}(v) \leq B - \deg_{H_{\init}}(u) < \left( \frac{1}{2} - \delta \right)B$, and so $v$ must be a low node.

Next, we prove \cref{overview:enu:low}. Let $e = (u, v)$, where both $u, v$ are low nodes. Then $\deg_{H_{\init}}(u) < \left(\frac{1}{2} - \delta\right)B$ and $\deg_{H_{\init}}(v) < \left(\frac{1}{2} - \delta  \right)B$, and hence $\deg_{H_{\init}}(u) + \deg_{H_{\init}}(v) < (1- 2\delta) B < (1-  \epsilon)B$. As $H_{\init}$ is a $(B, (1-\epsilon)B)$-EDCS of $G_{\init}$, it follows that $(u, v) \in E\left(H_{\init}\right)$ (see~\cref{defn:edcs}).
\end{proof}

We next introduce one last category of nodes, as defined below.

\begin{definition}
\label{overview:def:safe}
Consider a very high node $v \in V_{\vhi}$. At any point in time within the current phase, we say that $v$ is {\bf damaged} if  $\deg_{H_{\core}}(v) < \left( \frac{1}{2} + \delta - 2 \epsilon \right)B$, and {\bf safe} otherwise. We let $V_{\dmg} \subseteq V_{\vhi}$ and $V_{\safe} := V_{\vhi} \setminus V_{\dmg}$ respectively denote the sets of damaged and safe nodes.
\end{definition}

Recall that   $\deg_{H_{\init}}(v) > \left(\frac{1}{2} + \delta  \right)B$ for all  nodes $v \in V_{\vhi}$. Thus,  for  such a node to get damaged at least $2\epsilon B$ many of its incident edges must be deleted since the start of the phase. Since each phase lasts for $\delta n$ updates in $G$ and $\delta = 100 \epsilon$ (see~(\ref{overview:eq:delta})), we get the following important corollary. 

\begin{cor}
\label{overview:cor:damaged}
We always have $\left|V_{\dmg}\right| \leq \frac{2\delta n}{2\epsilon B} = O\left(\frac{n}{B} \right)$. Furthermore, at the start of a phase we have $V_{\dmg} = \emptyset$. During the phase, the subset $V_{\dmg}$ grows monotonically over time.
\end{cor}

\subsection{Maintaining a Matching of All the Safe Nodes}
\label{sec:overview:safe}

One of our key technical insights is that the degree gap in $H_{\core}$ allows us to efficiently maintain a matching $M_{\base} \subseteq E(H_{\core})$ that matches \underline{all} safe nodes, as per~\cref{overview:lm:base} below. With appropriate setting of parameters, the update time of both of the algorithms (deterministic and randomized) in~\cref{overview:lm:base} are sublinear in $n$.

\begin{lemma}
\label{overview:lm:base}
We can maintain a matching $M_{\base} 
\subseteq E(H_{\core})$ that satisfy the following properties.
\begin{enumerate}
\item \label{overview:enum:base:1} Every safe  node is matched  under $M_{\base}$, i.e., $V_{\safe} \subseteq V(M_{\base})$.
\item \label{overview:enum:base:2} Every update in $G$ that is internal to a phase (i.e., excluding those updates where one phase ends and the next phase begins) leads to at most $O(1)$  node insertions/deletions in $V(M_{\base})$.
\end{enumerate}
The matching $M_{\base}$ can be maintained either by
\begin{itemize}
\item a deterministic algorithm with $\tilde{O}\left(B \cdot n^{1/2} \cdot \delta^{-3/2}\right)$ amortized update time; or
\item  a randomized algorithm with $\tilde{O}\left(B \cdot \delta^{-1} + \delta^{-3} \right)$ amortized update time. The algorithm is correct with high probability against an adaptive adversary.
\end{itemize}
\end{lemma}

\medskip
\paragraph{Intuition Aside:} The only property of $H_{\core}$ that we need for Lemma \cref{overview:lm:base} is the degree gap between $V_{\safe}$ and $V_{\lo}$; in fact, the whole reason we use an EDCS in the first place is to ensure this degree gap. 

We defer the proof sketch of~\cref{overview:lm:base} to~\cref{sec:outline:proof}, but to provide an intuition for why the degree gap helps match vertices in $V_{\safe}$, we show here that it is  easy to deterministically maintain a \emph{fractional matching} satisfying the concerned properties of~\cref{overview:lm:base}, in $O(B/\delta)$  update time. We do not this result in the main body of our paper and include it only for intuition.
We denote the fractional matching by $x_{\base} : E(H_{\core}) \rightarrow [0,1]$, which is constructed as follows.
\begin{wrapper}
\begin{itemize}
\item $\Gamma \leftarrow \left(\frac{1}{2} + \delta - 2\epsilon\right)B$ and $x_{\base}(e) \leftarrow 0$ for all $e \in E$
\item For every safe node $v \in V_{\safe}$
\begin{itemize}
\item Let $E^{\star}_v \subseteq E(H_{\core})$ be a set of $\Gamma$ distinct edges incident to $v$. ($E^{\star}_v$ exists as $v \in V_{\safe}$).
\item $x_{\base}(e) \leftarrow 1/\Gamma$ for all $e \in E^{\star}_v$
\end{itemize}
\item Return $x_{\base} : E(H_{\core}) \rightarrow [0,1]$.
\end{itemize}
\end{wrapper}

Consider any two distinct nodes $v, v' \in V_{\safe} \subseteq V_{\vhi}$. By~\cref{overview:enu:very-high} of~\cref{overview:prop:fact}, there cannot be an edge $(v, v') \in E(H_{\core}) \subseteq E(H_{\init})$, and hence $E^\star_v \cap E^\star_{v'} = \emptyset$. Using this observation, it is easy to verify that the function $x_{\base}$ is a valid fractional matching in $H_{\core}$, and that every node $v \in V_{\safe}$  receives a weight  $x_{\base}(v) = 1$ under this fractional matching.\footnote{The weight $x_{\base}(v)$  is the sum of the values $x_{\base}(e)$, over all the edges in $H_{\core}$ that are incident on $v$.} Furthermore, since the EDCS $H_{\init}$ has maximum degree at most $B$ (this follows from~\cref{defn:edcs}) and each phase lasts for $\delta n$ updates, we can deterministically maintain the fractional matching $x_{\base}$ in $O\left(\frac{nB}{\delta n} \right) = O(B/\delta)$ amortized update time. One can also verify that each update in $G$ that is internal to a phase leads to at most $O(1)$ nodes $v \in V$ changing their weights $x_{\base}(v)$ under $x_{\base}$.

\thatchaphol{we will need to add that $E^\star_v$ for all $v$ can be maintained in $O(1)$ update time\aaron{Given that this is for intuition only I think it's fine to do nothing here.}}

\subsection{Maintaining a Maximal Matching in the Adjunct Subgraph}
\label{overview:sec:adjunct}

Let $G_{\res} := (V_{\res}, E_{\res})$ denote the subgraph of $G$ induced by the set of nodes that are unmatched under $M_{\base}$, that is,  $V_{\res} := V \setminus V(M_{\base})$ and  $E_{\res} := \{ (u, v) \in E : u, v \in V_{\res}\}$. We will refer to $G_{\res}$ as the {\bf adjunct subgraph} of $G$ w.r.t.~$M_{\base}$. By the following observation, all that remains is to maintain a maximal matching in $G_{\res}$

\begin{observation}
\label{overview:obs:mfinal}
If $M_{\res}$ is a maximal matching of $G_{\res}$, then
$M_{\fnl} := M_{\base} \cup M_{\res}$ is a maximal matching in  $G = (V, E)$
\end{observation}


We now show how to efficiently maintain a maximal matching $M_{\res}$ of $G_{\res}$ (we refer to it as the \textbf{adjunct matching}).
 
\begin{lem}
\label{overview:lem:maximal} Suppose that we  maintain the matching $M_{\base}$ as per~\cref{overview:lm:base}. Then with an additive overhead of $\tilde{O}\left(\frac{B}{\delta}+ \frac{n}{\delta B}\right)$ amortized update time, we can deterministically and explicitly maintain a \emph{maximal} matching $M_{\res} \subseteq E_{\res}$ in $G_{\res}$. 
\end{lem}

At this point, the reader might get alarmed because  $G_{\res}$ undergoes \emph{node-updates} every time an edge gets added to/removed from $M_{\base}$, which might indicate that it is impossible to efficiently maintain a maximal matching in $G_{\res}$. To assuage this concern, we emphasize that: (i) the {\em arboricity} of $G_{\res}$ is sublinear in $n$ (see~\cref{overview:arboricity}), and (ii) an edge update in $G$ that is internal to a phase leads to at most $O(1)$  node-updates in $G_{\res}$ (see~\cref{overview:arboricity:2}). These two properties ensure that we can maintain the matching $M_{\res}$ in sublinear in $n$ update time. 

\begin{claim}
    \label{overview:arboricity}
    The subgraph $G_{\res} = (V, E_{\res})$ satisfies the following properties.
    \begin{enumerate}
    \item \label{overview:enum:arb:1} $V_{\safe} \cap V_{\res} = \emptyset$ and $\left|V_{\dmg} \cap V_{\res} \right| = O\left(\frac{n}{B}\right)$. Furthermore, $V_{\dmg} \cap V_{\res} = \emptyset$ at the start of a phase.
    \item \label{overview:enum:arb:2} $\deg_{G_{\res}}(v) \leq B + \left| V_{\dmg} \cap V_{\res} \right|$ for all nodes $v \in V_{\lo}$.
    \end{enumerate}
\end{claim}

\begin{proof}
Since $V_{\res} := V \setminus V(M_{\base})$,~\cref{overview:enum:arb:1} follows from~\cref{overview:enum:base:1} of~\cref{overview:lm:base}, and~\cref{overview:cor:damaged}.

To prove~\cref{overview:enum:arb:2}, consider any node $v \in V_{\lo}$. Let $u \in V_{\lo}$ be any low neighbor of $v$ in $E_{\res} \subseteq E \subseteq E_{\init}$. Then, by~\cref{overview:enu:low} of~\cref{overview:prop:fact}, we   have $(u, v) \in H_{\init}$. Accordingly, the number of low neighbors of $v$ in $G_{\res}$ is at most $\deg_{H_{\init}}(v) \leq B$, where the last inequality holds because $H_{\init}$ is a $(B, (1-\epsilon)B)$-EDCS of $G_{\init}$ (see~\cref{defn:edcs}). As the node-set $V$ is partitioned into  $V_{\lo}, V_{\vhi}$ (see~\cref{assume:gap}) and  $V_{\vhi}$ is further partitioned into $V_{\safe}, V_{\dmg}$, we get: $$\deg_{G_{\res}}(v) \leq B + |V_{\safe} \cap V_{\res}| + |V_{\dmg} \cap V_{\res}| =  B + |V_{\dmg} \cap V_{\res}|,$$ 
where the last equality follows from~\cref{overview:enum:arb:1}. 
\end{proof}

\begin{claim}
    \label{overview:arboricity:2} An edge-update in $G$  internal to a phase leads to at most $O(1)$  node-updates in $G_{\res}$.
\end{claim}

\begin{proof}
Since $V_{\res} := V \setminus V(M_{\base})$, the claim follows from~\cref{overview:enum:base:2} of~\cref{overview:lm:base}.
\end{proof}

\begin{corollary}
    \label{overview:cor:arboricity} At the start of a phase, we have $|E_{\res}| = O(Bn)$.
\end{corollary}

\begin{proof}

By \Cref{overview:prop:fact}(\ref{overview:enu:low}) we have that $E_\res \subseteq H_\init$. As $H_\init$ is a $(B, B(1-\epsilon))$-EDCS and the maximum vertex degree of a $(B, B(1-\epsilon))$-EDCS is $B$ we must have that $|E_\res| \leq |H_\init|  = O(nB)$.

\end{proof}

\subsubsection*{Proof of~\cref{overview:lem:maximal}}

At the start of a phase, we initialize $M_{\res}$ to be any arbitrary maximal matching in $G_{\res}$. By~\cref{overview:cor:arboricity}, this takes $O(Bn)$ time. Since the phase lasts for $\delta n$ updates in $G$, this step incurs an amortized update time of:
\begin{equation}
\label{overview:eq:update:time:0}
O\left(\frac{B n}{\delta n}\right) = O\left(\frac{B}{\delta}\right).
\end{equation}

Subsequently, to handle the updates during a phase, we  use the following auxiliary data structures: Each node $v \in V_{\dmg}$ maintains the set $\F_{\res}(v) := \left\{ u \in V_{\res} \setminus V(M_{\res}) : (u, v) \in E_{\res}\right\}$ of its free neighbors (in $G_{\res}$) under the matching $M_{\res}$.\footnote{Maximality of $M_{\res}$ implies that $\F_{\res}(v) = \emptyset$ for all $v \in V_{\dmg} \setminus  V(M_{\res})$.} Whenever a node $v$ moves from $V_{\safe}$ to $V_{\dmg}$, we spend $\tilde{O}(n)$ time to initialize the set $\F_{\res}(v)$ as a balanced search tree. Within a phase, a  node can move from $V_{\safe}$ to $V_{\dmg}$ at most once (see~\cref{overview:def:safe} and~\cref{overview:cor:damaged}). Thus, by \cref{overview:enum:arb:1} of~\cref{overview:arboricity}, these initializations take $\tilde{O}\left( \frac{n}{B} \cdot n \right) = \tilde{O}\left(\frac{n^2}{B} \right)$ total time during a phase. As each phase lasts for $\delta n$ updates in $G$, this incurs an amortized update time of $\tilde{O}\left( \frac{n^2}{B \cdot \delta n}  \right) = \tilde{O}\left( \frac{n}{\delta B}\right)$.

\cref{overview:arboricity:2} guarantees that each edge update in $G$ leads to $O(1)$ node-updates in $G_{\res}$. We now show that to maintain the maximal matching $M_{\res}$, a node-update to $G_{\res}$ can be handled in $\tilde{O}\left(B + \frac{n}{B} \right)$ worst-case time.
If a new vertex $v$ enters $G_{\res}$ (because $v$ became unmatched in $M_{\base}$), then to maintain the maximality of $M_{\res}$, the algorithm needs to find a free neighbor of $v$ (if one exists). The same needs to be done if $x$ leaves $G_{\res}$, and $x$ was previously matched to $v$ in $M_{\res}$. So to handle a node-update in $G_{\res}$ we need a subroutine that takes in a vertex $v$ and either finds a $M_{\res}$-free neighbor of $v$ in $G_\res$ or certifies that none exists.

The above subroutine can be performed in $\otil(1)$  time if $v \in V_{\dmg}$ (using the set $\F_{\res}(v)$); if $v \in V_{\res} \setminus V_{\dmg}$ then by \cref{overview:enum:arb:1} of \cref{overview:arboricity} $v \in V_{\lo}$, so by \cref{overview:enum:arb:2} of the same claim we can perform the subroutine in time  $\deg_{G_{\res}}(v) = O\left(B + \frac{n}{B} \right)$.
Further, once we decide to match a node $u \in V_{\res}$ to one of its free neighbors, we can  update all the relevant sets $\F_{\res}(v)$ in $\tilde{O}\left( \left| V_{\dmg}\right| \right) = \tilde{O}\left( \frac{n}{B} \right)$ time (see~\cref{overview:arboricity}).

To summarize, the amortized update of maintaining the matching $M_{\res}$ under adversarial updates in $G$ is at most:
\begin{equation}
\label{overview:eq:update:time:2}
\tilde{O}\left(\frac{n}{\delta B}  + B + \frac{n}{B} \right) = \tilde{O}\left(B + \frac{n}{\delta B} \right).
\end{equation}

\cref{overview:lem:maximal} follows from~(\ref{overview:eq:update:time:0}) and~(\ref{overview:eq:update:time:2}).

\subsection{Putting Everything Together: How to Derive~\cref{thm:main}}
\label{sec:together}

Since $E(H_{\core}) \subseteq E$, it follows that $M_{\base} \subseteq E(H_{\core})$ is a valid matching in $G$ (see~\cref{overview:lm:base}). Furthermore,~\cref{overview:lem:maximal} guarantees that $M_{\res}$ is a \emph{maximal} matching in $G_{\res}$
Thus, by \cref{overview:obs:mfinal}, the matching $M_{\fnl} := M_{\base} \cup M_{\res}$  is  a maximal matching in $G$. It now remains to analyze the overall amortized update time of our algorithm.

\medskip
Recall that maintaining the EDCS $H$ incurs an update time of $O\left( \frac{n}{\epsilon B}\right) = O\left( \frac{n}{\delta B}\right)$, according to~(\ref{overview:eq:delta}). Now, for the deterministic algorithm,~\cref{overview:lm:base} incurs an amortized update time of $\tilde{O}\left( B \cdot n^{1/2} \cdot \delta^{-3/2}\right)$.  Thus, by~\cref{overview:lem:maximal}, the overall amortized update time becomes
$$O\left(\frac{n}{\delta B}\right) + \tilde{O}\left( B \cdot n^{1/2} \cdot \delta^{-3/2}\right) + \tilde{O}\left(\frac{B}{\delta} + \frac{n}{\delta B}\right)  = \tilde{O}(n^{8/9}), \text{ for } \delta = \frac{1}{n^{1/9}} \text{ and } B = n^{2/9}.$$

For the randomized algorithm against an adaptive adversary,~\cref{overview:lm:base} incurs an amortized update time of $\tilde{O}(B \delta^{-1} + \delta^{-3})$. Thus, by~\cref{overview:lem:maximal}, the overall amortized update time becomes
$$O\left(\frac{n}{\delta B}\right) + \tilde{O}\left(\frac{B}{\delta}+ \frac{1}{\delta^{3}}\right) + \tilde{O}\left(\frac{B}{\delta}+ \frac{n}{\delta B}\right) = \tilde{O}\left( n^{3/4} \right), \text{ for } \delta = \frac{1}{n^{1/4}} \text{ and } B = n^{1/2}.$$

This leads us to~\cref{thm:main} in the decremental setting.

\subsection{Dealing with Medium Nodes: Getting Rid of~\cref{assume:gap}}
\label{sec:assume:gap}

We now provide a high-level outline of our algorithm in the general case, when~\cref{assume:gap} does not hold. Recall that under~\cref{assume:gap} there cannot be any node $v$ with $\left(\frac{1}{2} - \delta\right)B \leq  \deg_{H_{\init}}(v) \leq \left(\frac{1}{2} + \delta  \right)B$. Conceptually, the most significant challenge here is to deal with the set of {\bf medium nodes}, given by: $$V_{\med} := \left\{ v \in V : \left(\frac{1}{2} - \delta\right)B \leq  \deg_{H_{\init}}(v) < \left(\frac{1}{2} + \delta - \epsilon \right)B \right\}.$$

For the rest of this overview, we replace \cref{assume:gap} with the following weaker assumption

\begin{assumption}[Replaces \cref{assume:gap}]
\label{assume:medium}
All vertices are in $V_{\lo}$, $V_{\med}$, or $V_{\vhi}$
\end{assumption}

In the main body of the paper we also have to handle the nodes  that are ruled out by the above assumption, namely the ones whose degrees in $H_{\init}$ lie in the range $\left[\left(\frac{1}{2} + \delta - \epsilon\right)B, \left(\frac{1}{2}+\delta\right)B \right]$. These remaining nodes lie in an extremely narrow range, so while they require some technical massaging, they do not pose any additional conceptual difficulties.

\medskip
We now show how to adapt our algorithm to also handle the medium nodes. \cref{overview:lm:base:general}, stated below, generalizes~\cref{overview:lm:base}, and serves as a key building block of our algorithm. The only difference between these two lemmas is that in~\cref{overview:enum:base:general:1} of~\cref{overview:lm:base:general}, we allow for $O(\delta n)$ medium nodes that are free under $M_{\base}$. Note that sets $V_{\safe}$ and $V_{\dmg}$ are defined the same as before (\Cref{overview:def:safe}), and are in particular both subsets of $V_{\vhi}$.

\begin{lemma}
\label{overview:lm:base:general}
We can maintain a matching $M_{\base} 
\subseteq E(H_{\core})$ that satisfy the following properties.
\begin{enumerate}
\item \label{overview:enum:base:general:1} Every safe  node is matched  under $M_{\base}$, i.e., $V_{\safe} \subseteq V(M_{\base})$. Furthermore, at most $O(\delta n)$ medium nodes are free under $M_{\base}$, i.e., $|V_{\med} \setminus V(M_{\base})| = O(\delta n)$. 
\item \label{overview:enum:base:general:2} Every update in $G$ that is internal to a phase (i.e., excluding those updates where one phase ends and the next phase begins) leads to at most $O(1)$  node insertions/deletions in $V(M_{\base})$.
\end{enumerate}
The matching $M_{\base}$ can be maintained
\begin{itemize}
\item either by a deterministic algorithm with $\tilde{O}\left(B \cdot n^{1/2} \cdot \delta^{-3/2}\right)$ amortized update time; or
\item by a randomized algorithm with $\tilde{O}\left(B \cdot \delta^{-1} + \delta^{-3} \right)$ amortized update time. The algorithm is correct with high probability against an adaptive adversary.
\end{itemize}
\end{lemma}

Before outlining the proof sketch of~\cref{overview:lm:base:general}, we explain how to adapt the algorithm for maintaining a maximal matching $M_{\res} \subseteq E_{\res}$ in the adjunct graph $G_{\res} := (V_{\res}, E_{\res})$, which is the subgraph of $G$ induced by the node-set $V_{\res} := V \setminus V(M_{\base})$ (see~\cref{overview:sec:adjunct}).~\cref{overview:lem:maximal:general}, stated below, generalizes~\cref{overview:lem:maximal} in the presence of medium nodes.

\begin{lem}
\label{overview:lem:maximal:general} Suppose that we  maintain the matching $M_{\base}$ as per~\cref{overview:lm:base:general}. Then with an additive overhead of $\tilde{O}\left(\delta n + \frac{B}{\delta}+ \frac{n}{\delta B}\right)$ amortized update time, we can deterministically and explicitly maintain a \emph{maximal} matching $M_{\res} \subseteq E_{\res}$ in $G_{\res}$. 
\end{lem}

\begin{proof}[Proof Sketch]
The only difference with the proof of~\cref{overview:lem:maximal} is that now the set $V_{\res}$ can contain some medium nodes, although $|V_{\res} \cap V_{\med}| = O(\delta n)$ as per~\cref{overview:enum:base:general:1} of~\cref{overview:lm:base:general}. 

We explicitly maintain the set  $V_{\res} \cap V_{\med}$ as a doubly linked list. Whenever a node $v \in V_{\lo} \cup V_{\dmg}$ is searching for a free neighbor (in the proof of~\cref{overview:lem:maximal}), we ensure that it additionally scans the set $V_{\res} \cap V_{\med}$. This leads to an additive overhead of~$|V_{\med} \cap V_{\res}| = O(\delta n)$ in the update time.

Next, recall that $H_{\init}$ is a $(B, (1-\epsilon)B)$-EDCS of $G_{\init}$. Consider any edge $(u, v) \in E_{\res} \subseteq E \subseteq E_{\init}$ with  $v \in V_{\med}$ and $u \in V_{\lo}$.  Then, we have $\deg_{H_{\init}}(u) + \deg_{H_{\init}}(v) < \left(\frac{1}{2} - \delta\right)B + \left( \frac{1}{2} + \delta - \epsilon \right)B = (1-\epsilon)B$, and hence~\cref{defn:edcs} implies that $(u, v) \in E(H_{\init})$. In other words, all the edges in $G_{\res}$ that connect a low node with a medium node belong to the EDCS $H_{\init}$, which in turn has maximum degree at most $B$ (see~\cref{defn:edcs}). Thus, every medium node has at most $B$ low neighbors in $G_{\res}$. 

Now, suppose that we are searching for a free neighbor of a node $v \in V_{\med}$ in $G_{\res}$. This step involves scanning through: (i) all the low neighbors of $v$ in $G_{\res}$ (and by the previous discussion there are at most $B$ such neighbors), (ii) the set $V_{\med} \cap V_{\res}$ of size $O(\delta n)$, and (iii) and at most $O\left(\frac{n}{B}\right)$ damaged nodes $V_{\dmg} \cap V_{\res}$ (see~\cref{overview:sec:adjunct}). Overall, this scan takes $O\left(\delta n + B + \frac{n}{B}\right)$ time, and just as before, contributes an additive overhead of $O(\delta n)$ to the update time in~\cref{overview:lem:maximal}.

Everything else remains the same as in the proof of~\cref{overview:lem:maximal}.
\end{proof}

Comparing~\cref{overview:lm:base:general,overview:lem:maximal:general} against~\cref{overview:lm:base,overview:lem:maximal}, we conclude that there is an additive overhead of $O(\delta n)$ update time while dealing with the medium nodes. Finally, from the analysis from~\cref{sec:together}, it is easy to verify that this  overhead of $O(\delta n)$ does \emph{not} degrade our overall asymptotic update time (the values of $\delta, B$ remain the same as in~\cref{sec:together}).

\subsubsection*{Proof Sketch of~\cref{overview:lm:base:general}}

Consider the scenario at the start of a new phase (just before the first update within that phase). At this point, we first compute a subgraph $H'_{\init} := (V, E(H'_{\init}))$ of $H_{\init}$, as follows. We initialize $H'_{\init} \leftarrow H_{\init}$. Subsequently, for every node $v \in V_{\vhi}$, we keep deleting edges $(u, v) \in E(H'_{\init})$ incident on $v$ from $H'_{\init}$, until $\deg_{H'_{\init}}(v)$ becomes equal to $\Delta := \left(\frac{1}{2} +\delta \right)B$. By~\cref{overview:enu:very-high} of~\cref{overview:prop:fact}, every edge $(u, v)$ that gets deleted from $H'_{\init}$  has its other endpoint $u \in V_{\lo}$. Accordingly, this process does not reduce the degree of any medium node. When this for loop terminates, we have $\left( \frac{1}{2} - \delta \right)B \leq \deg_{H'_{\init}}(v) \leq \Delta$ for all $v \in V_{\med} \cup V_{\vhi}$, and $\deg_{H'_{\init}}(v) < \left(\frac{1}{2} - \delta\right)B$ for all $v \in V_{\lo}$. In other words, the degree of every node $v \in V_{\med} \cup V_{\vhi}$ in $H'_{\init}$ is at least $(1-\Theta(\delta))$ times the maximum degree $\Delta$. Now, we apply~\cref{lem:match almost max degree}, with $\kappa = \Theta(\delta)$, to compute a matching $M \subseteq E(H'_{\init}) \subseteq E(H_{\init})$ that matches all but $O(\delta n)$ nodes in $V_{\med} \cup V_{\vhi}$. Note that $\Delta = \Theta(B)$, and that a phase lasts for $\delta n$ updates in $G$. Thus, at the start of a phase, computing the subgraph $H'_{\init}$ and the invocation of~\cref{lem:match almost max degree}  incurs an amortized update time of $$\tilde{O}\left( \frac{nB + n \Delta/\delta}{\delta n} \right) = \tilde{O}\left( \frac{B}{\delta^2} \right)$$ if we use the deterministic algorithm, and an amortized update time of $$\tilde{O}\left( \frac{nB + n \Delta \delta}{\delta n} \right) = \tilde{O}\left( \frac{B}{\delta} \right)$$ if we use the randomized algorithm. It is easy to verify that both these terms get subsumed within the respective (deterministic or randomized)  guarantees of~\cref{overview:lm:base:general}. 

We have thus shown to initialize the desired matching $M_{\base} = M$ at the beginning of the phase. 
We now give a sketch of how we maintain $M_{\base}$ dynamically throughout the phase.
Define $H_{\hilo}$ to be the induced graph $H_{\core}[V_{\vhi} \cup V_{\lo}]$. Recall that all edges in $E_{H_{\core}}(V_{\vhi})$ have their other endpoint in $V_{\lo}$ (\cref{overview:enu:very-high} of \cref{overview:prop:fact}), so the set of edges incident to $V_{\vhi}$ is exactly the same in $E_{H_{\hilo}}$ and $E_{H_{\core}}$. 
We are thus able to apply \cref{overview:lm:base} exactly as before to maintain a matching $M_{\hilo}$ in $H_{\hilo}$ that matches all of $V_{\safe}$. We will maintain the invariant that $M_{\hilo} \subseteq M_{\base}$, thus guaranteeing that all nodes in $V_{\safe}$ remain matched throughout the phase. The only remaining concern is that the algorithm for maintaining $M_{\hilo}$ might end up unmatching medium vertices in $M_{\base}$: if vertex $u \in V_{\hi}$ becomes unmatched, then to maintain $M_{\hilo}$ the algorithm from \cref{overview:lm:base} finds an augmenting path from $u$ to some $v \in V_{\lo}$ that was free in $M_{\hilo}$; but although $v$ was free in $M_{\hilo}$, it might have been matched to some medium node $x$ in $M_{\base}$, in which case node $x$ now becomes unmatched. Fortunately, it is not hard to show that every adversarial update creates only $O(1)$ free vertices, so since $M_{\base}$ has $O(\delta n)$ free medium vertices at the beginning of the phase, and since each phase lasts for $O(\delta n)$ updates, we maintain the desired property that there are always $O(\delta n)$ free medium vertices.

\subsection{Matching All the Safe Nodes: Proof Outline for~\cref{overview:lm:base}}
\label{sec:outline:proof}

The key to our proof is an algorithm for maintaining a left-perfect matching in a bipartite graph with a degree gap.

\begin{lemma}[Simplified version of \cref{lem:perfect matching}]
\label{lm:overview:degree-gap}
Let $G = (L \cup R, E)$ be a bipartite graph where, for some $X > 0$ and $0 < \gamma < 1$, every vertex $v \in L$ has $2X \geq \deg(v) \geq X$ and every vertex in $R$ has degree $\leq X(1-\gamma)$. Suppose that the above degree constraints always hold as $G$ is subject to $O(\delta n)$ decremental updates. Then, there exists an algorithm that maintains a left-perfect matching in $G$ -- i.e. a matching in which every vertex in $L$ is matched -- with the following update times: 
\begin{itemize}
\item A randomized algorithm with worst-case update time $\otil\left(X + 1/\gamma^3\right)$.
\item A deterministic algorithm with \emph{total} update time $\otil\left(\frac{nX}{\gamma} + \frac{n^{1.5}X\delta}{\gamma^{1.5}}\right)$.
\end{itemize}
\end{lemma}

Before proving the above Lemma, let us see why it implies \cref{overview:lm:base}

\begin{proof}[Proof sketch of \cref{overview:lm:base}]
For this proof sketch, we assume that we always have $V_{\vhi} = V_{\safe}$; in the full version, we are able to effectively ignore damaged nodes because  \cref{overview:lm:base} only requires us to match all the nodes of $V_{\safe}$. Let $H_{\hilo}$ be the subgraph of $H_{\core}$ that contains only the edges incident to $V_{\vhi}$. Observe that under our assumption of $V_{\vhi} = V_{\safe}$, $H_{\hilo}$ satisfies the properties of \cref{lm:overview:degree-gap} with $X = B/2$ and $\gamma \sim \delta$: note in particular that the first property of an EDCS (\cref{defn:edcs}) ensures that every vertex has degree $\leq B = 2X$, and that $H_{\hilo}$ is bipartite because by \cref{overview:prop:fact}, $H_{\init}$ contains no edges between vertices in $V_{\vhi}$.
We can thus apply \cref{lm:overview:degree-gap} to maintain a matching that matches all of $V_{\safe}$, as desired. The update-time bounds of  \cref{lm:overview:degree-gap} precisely imply those of \cref{overview:lm:base}\thatchaphol{For randomized version, the running time does not match.}; for the deterministic algorithm, recall that we amortize over the $O(\delta n)$ updates of a phase.

\end{proof}


\subsubsection{Maintaining Left-Perfect Matching: Proof of \cref{lm:overview:degree-gap}}

\paragraph{Randomized Algorithm.} The randomized algorithm is extremely simple. Say that we have a matching $M$ under which some $v \in L$ is still unmatched. We then find an augmenting path from $v$ to an $M$-free node in $R$ by taking a \emph{random matching-walk} from $v$. That is, we follow a random edge $(v,x)$ to $R$; if $x$ is unmatched then we have successfully found an augmenting path, while if $x$ is matched to some $y \in L$ then we repeat the process from $y$. 

We are able to show that because of the assumed degree gap of $G$, this random matching-walk terminates within $\otil(1/\gamma^{3})$ steps with high probability. The intuitive reason is as follows:  because of the degree gap, every set $S \subset L$ as at least $S(1+\gamma)$ neighbors in $R$. The graph is thus effectively a $\gamma$-expander, and in particular we are able to show a correspondence between random matching-walks in $G$ and ordinary random walks in some Eulerian directed $\gamma$-expander $G'$; the bound of  $\otil(1/\gamma^{3})$ then follows from known facts about Eulerian expanders.

To initialize the a left perfect matching before the first decremental update occurs we apply the randomized algorithm of Lemma~\ref{lem:match almost max degree} with $\Delta = X, \kappa = \gamma$. This returns a matching covering all but $O(n\gamma)$ vertices of $L$ in time $\tilde{O}(nX\gamma)$ adding an amortized update time component of $\tilde{O}(X)$ over the deletions. We may match each of the $O(n\gamma)$ unmatched vertices using the same random walk based technique as we use to handle deletions in time $\tilde{O(n\gamma^{-2})}$ adding an amortized update time component of $\tilde{O(\gamma^{-3})}$ over the $O(n\gamma)$ deletions.

\paragraph{Deterministic Algorithm.}
The deterministic algorithm is more involved than the randomized one, so we only give a high-level sketch. We start by maintaining a standard directed residual graph $\gres$ corresponding to the matching $M$ that we maintain (edges in $E \setminus M$ point from $L$ to $R$ in $\gres$, while edges in $M$ point from $R$ to $L$). We also create a dummy sink $t$ in $\gres$, with an edge from all $M$-free vertices in $R$ to $t$. It is not hard to check that there is a correspondence between paths from $L$ to $t$ in $\gres$ and augmenting paths in $G$. Our algorithm maintains a shortest path tree $\tres$ to sink $t$ in $\gres$. Given any $M$-free free vertex $v \in L$, we can find an augmenting path from $v$ by simply following the $v-t$ path in $\tres$.

The primary challenge is to maintain $\tres$. It is not hard to prove that 
for any $M$-free $v \in L$, there always \emph{exists} an augmenting path of length $\otil(1/\gamma)$ from $v$ to an $M$-free vertex in $R$; this follows from the same expander-based arguments as above, and can also be proved using more elementary techniques. This means that the depth of $\tres$ is bounded by $\otil(1/\gamma)$, which makes it easier to maintain. 

There exists a classic data structure known as an \textit{Even and Shiloach tree} (denoted ES-tree) that can maintain a low-depth shortest path tree to a sink $t$ in a decremental graph \cite{EvenS81}. The ES-tree can also handle a special kind of edge insertion into $G$ -- namely, an edge insertion is a valid update if it does not decrease any distances to $t$. If $\gres$ were a decremental graph, then we could us an ES-tree to maintain $\tres$, and we would be done. Unfortunately, even though the input graph $G$ to 
\cref{lm:overview:degree-gap} only undergoes edge deletions, the graph $\gres$ can undergo edge insertions for two reasons.

  


Firstly, whenever we augment down a path $P$, every edge in $P$ flips whether or not it is in $M$, which means we need to flip the directions  of all these edges in $\Gres$. We can model this flipping as an inserting of all edges on the reverse path $\rev{P}$, followed by a deletion of all edges on the original path $P$. Fortunately, it is not hard to show that because $\tres$ is a shortest-path tree, inserting $\rev{P}$ does not decrease any distances (as long as we so before deleting $P$), and hence is a valid update to the ES-tree. 

It is the second kind of edge insertion to $\gres$ that poses a significant problem. Say that the adversary deletes some matching edge $(u,v)$ in $G$. Then $v$ becomes free, so the graph $\gres$  should now contain an edge $(v,t)$; but simply inserting this edge into $\gres$ would not be a valid update to the ES-tree, as it would decrease $\dist(v,t)$. To overcome this issue, we observe that before the deletion of edge $(u,v)$ we had $\dist_{\gres}(v,t) = \otil(1/\gamma)$, because we could follow the matching edge from $v$ to $u$, and as discussed above we know that there exists an augmenting path from $u$ to a free vertex of length $\otil(1/\gamma)$. So inserting an edge $(v,t)$ of weight $\otil(1/\gamma)$ does not decrease any distances to $t$ and is hence a valid operation for the ES-tree. Unfortunately, once we start adding weights to edges to of $\gres$, the distances in $\gres$ will get longer and longer, so the next time we add an edge $(v',t)$ it might need to have even larger weight for this to be a valid insertion into the ES-tree. (In particular, the presence of edge weights means that we can no longer guarantee  that before the insertion of $(v',t)$ we had $\dist(v',t) = \otil(1/\gamma)$.) The distances in $\gres$ can thus blow up, which increases the (weighted) depth of $\tres$ and hence causes  the ES-tree to become inefficient.

To overcome the above issue, we define parameter $q = \sqrt{n\gamma}$ and use structural properties of degree-gap graphs to show that as long as $\leq q$ insertions have occurred, the edge weights have not had a chance to get too large, and the \emph{average} $\dist(v,t)$ remains small, so the ES-tree remains efficient. Then, after $q$ updates, we simply rebuild the entire ES tree from scratch. Since the total number of resets is only $O(n\delta/q)$, the total update time ends up being not too large.

To initialize a left perfect matching we will use the deterministic algorithm of Lemma~\ref{lem:match almost max degree} which returns a matching covering all but $O(n\gamma)$ vertices of $L$ in time $\tilde{O}(nX/\gamma)$. We augment this imperfect matching at initialization to be a perfect matching using the same datastructure as we use to handle edge deletions. This results in a running time proportional to handling $O(n\gamma)$ deletions matching the cost of handling the decremental updates.




\subsection{Extension to the Fully Dynamic Setting}
\label{sec:fullydynamic:overview}

It is relatively straightforward to extend the algorithm from~\cref{sec:decremental:overview} to a setting where the input graph $G = (V, E)$ undergoes both edge insertions and edge deletions. As in~\cref{sec:decremental:overview}, our fully dynamic algorithm works in { phases}, where each phase lasts for $\delta n$ updates in $G$. 

We maintain two subgraphs $G_{\texttt{new}} := (V, E_{\texttt{new}})$ and $G_{\core} := (V, E_{\core})$ of $G = (V, E)$, such that the edge-set $E$ is partitioned into $E_{\texttt{new}}$ and $E_{\core}$. To be more specific, at the start of a phase, we set $E_{\texttt{new}} := \emptyset$ and $E_{\core} := E$. Subsequently, during the phase,   whenever an edge $e$ gets deleted from $G$, we set $E_{\core} \leftarrow E_{\core} \setminus \{e\}$ and $E_{\texttt{new}} \leftarrow E_{\texttt{new}} \setminus \{e\}$. In contrast, whenever an edge $e$ gets inserted into $G$, we leave the set $E_{\core}$ unchanged, and set $E_{\texttt{new}} \leftarrow E_{\texttt{new}} \cup \{e\}$. 

It is easy to verify that within any given phase $G_{\core}$ is a {\em decremental} subgraph of $G$. Furthermore, since $E_{\texttt{new}} = \emptyset$ at the start of a phase, each update adds at most one edge to $E_{\texttt{new}}$ and a phase lasts for $\delta n$ updates, it follows that $|E_{\texttt{new}}| \leq \delta n$ at all times.

We maintain a maximal matching $M_{\core} \subseteq E_{\core}$ in  $G_{\core}$, using the decremental algorithm from~\cref{sec:decremental:overview}. In a bit more detail, we handle an update in $G$ as follows. We first let the decremental algorithm from~\cref{sec:decremental:overview} handle the resulting update in $G_{\core}$ and accordingly modify the matching $M_{\core}$. Next, we scan through all the edges in $G_{\texttt{new}}$ and greedily add as many of these edges as we can on top of $M_{\core}$, subject to the condition that the resulting edge-set $M_{\fnl} \supseteq M_{\core}$ remains a valid matching in $G$. This scan takes $O(|E_{\texttt{new}}|) = O(\delta n)$ time. It is easy to verify that the matching $M_{\fnl}$ we obtain at the end of our scan is a maximal matching in $G$.

To summarize, with an additive overhead of $O(\delta n)$ update time, we can convert the decremental maximal matching algorithm from~\cref{sec:decremental:overview} into a fully dynamic algorithm. Finally, looking back at the analysis from~\cref{sec:together}, it is easy to verify that this additive overhead of $O(\delta n)$ does \emph{not} degrade our overall asymptotic update time (the values of $\delta, B$ remain the same as in~\cref{sec:together}).

\section{Decremental Perfect Matching in Graphs with Degree Gap}
\label{sec:decremental}

We say that a bipartite graph $G=(L,R,E)$ has a \emph{$\gamma$-degree-gap at $X$} if $\deg_{G}(v)\ge X$ for all $v\in L$ and $\deg_{G}(v)\le X(1-\gamma)$ for all $v\in R$. Note that if $\gamma\ge0$, then, by Hall's theorem, there exists a matching $M$ that is \emph{left-perfect}, i.e., $M$ matches every vertex in $L$. 

\begin{lem}
\label{lem:perfect matching}Let $G=(L,R,E)$ be a bipartite graph with initial $\gamma$-degree-gap at $X$. We can build a data structure $\LPM$ (stands for Left Perfect Matching) on $G$ that maintains a matching $M$ and supports the following operations: 
\begin{itemize}
\item $\mathrm{Init}(M_{0})$ where $M_{0}$ is a matching in $G$ before any update: $M\gets M_{0}$. 
\item $\mathrm{Delete}(u,v)$ where $(u,v)\in E(L,R)$: set $G\gets G-(u,v)$ and $M\gets M-(u,v)$. If $\deg_{G}(u)<X$ and $(u,v')\in M$, $M\gets M-(u,v')$.
\item $\mathrm{Augment}(u)$ where $u\in L$ is $M$-free and $\deg_{G}(u)\ge X$: augment $M$ by an augmenting path so that $u$ is $M$-matched. 
\end{itemize}
The data structure can be implemented by either: 
\begin{enumerate}
\item \label{enu:det}a deterministic algorithm with $O(n\Delta\frac{\log(n)}{\gamma})\cdot(1+\frac{u}{\sqrt{n\gamma}})$ total update time where $\Delta$ is the maximum degree of $G$ and $u$ is the number of calls to $\mathrm{Delete}(\cdot)$. In particular, the total update time is independent of the number of calls to $\mathrm{Augment}(\cdot)$. 
\item \label{enu:rand}a randomized algorithm that takes $O(|M_{0}|)$ time to $\mathrm{Init}(\cdot)$, $O(X\log n)$ worst-case time to $\mathrm{Delete}(\cdot)$, and $O(\log^{3}(n)/\gamma^{3})$ time to $\mathrm{Augment}(\cdot)$. The algorithm is correct with high probability against an adaptive adversary. The algorithm assumes that each vertex in $G$ has access to a binary search tree of edges incident to it. 
\end{enumerate}
\end{lem}

Note that the $\LPM$ data structure allows one to maintain a matching that matches all nodes with $L$ with $\deg_G(u) \geq X$: whenever such a node $u$ becomes unmatched, call $\Augment(u)$. That is essentially how we will use the data structure in Section \ref{new:dyn:maximal}, but for convenience of the interface it helps us to have $\Augment$ as a separate procedure.

\subsection{Structural Properties of Degree-Gap Graphs}


In this section, we present the structural properties of graphs with a degree gap that are crucial for maintaining left-perfect matching in both our randomized and deterministic algorithms.

\begin{lem}
\label{lem:short-path-deg-gap} Let $G=(L,R,E)$ be a bipartite graph with $\gamma$-degree-gap at $X$. Let $M$ be some matching in $G$. Given any vertex $u\in L$ (matched or unmatched), there exists an alternating path $P$ from $u$ that starts with an edge $e\notin M$, ends at an $M$-free vertex, and has length $O(\log(n)/\gamma)$. 

Moreover, we can find such path $P$ in $O(\log^{3}(n)/\gamma^{3})$ time with high probability. This assumes that each vertex in $G$ has access to a binary search tree of edges incident to it.
\end{lem}
To prove the above lemma, we will analyze the graph   $G_M$ obtained from $G$ by (1) contracting each matched edge $(u_{\ell},u_{r})\in M$ where $u_{\ell}\in L$ and $u_{r}\in R$ into a single vertex $u$ and (2) contracting all free vertices into a single \emph{sink} vertex $t$. We keep parallel edges and remove self-loops from $G_{M}$. The crucial structure about $G_{M}$ is summarized below.
\begin{lem}
\label{lem:structure GM}For every non-sink vertex $u\in V(G_{M})$, 
\begin{enumerate}
\item \label{enu:exist path}There exists a path from $u$ to $t$ in $G_{M}$ of length $O(\log(n)/\gamma)$, and 
\item \label{enu:random walk}a random walk from $u$ in $G_{M}$ for $\Theta(\log(n)/\gamma^{2})$ steps passes through $t$ with probability $\Omega(\gamma)$. 
\end{enumerate}
\end{lem}

Before proving \Cref{lem:structure GM}, we show why it immediately implies \Cref{lem:short-path-deg-gap}.

\paragraph{Proof of \Cref{lem:short-path-deg-gap}.}
\begin{proof}
Each non-sink vertex $u$ in $G_{M}$ corresponds to a matched edge $(u_{\ell},u_{r})\in M$ where $u_{\ell}\in L$ and $u_{r}\in R$. Observe that any path 
\[
P_{M}=(u^{(1)},u^{(2)},u^{(3)},\dots,u^{(z)},t)
\]
from $u^{(1)}$ to $t$ in $G_{M}$ corresponds to an alternating path 
\[
P=(u_{\ell}^{(1)},u_{r}^{(2)},u_{\ell}^{(2)},u_{r}^{(3)},u_{\ell}^{(3)},\dots,u_{r}^{(z)},u_{\ell}^{(z)},v_{r})
\]
in $G$ where $v_{r}\in R$ is $M$-free and the first edge $(u_{\ell}^{(1)},u_{r}^{(2)})\notin M$ 

Therefore, \Cref{lem:structure GM}(\ref{enu:exist path}) implies that every $M$-matched vertex $u_{\ell}\in L$ has an alternating path $P$ of length $O(\log(n)/\gamma)$ from $u$ that starts with an edge $e\notin M$ and ends at an $M$-free vertex. 

For each $M$-free vertex $u_{\ell}\in L$, if $u_{\ell}$ has an $M$-free neighbor $v_{r}$, then $(u_{\ell},v_{r})$ is the desired path of length one. Otherwise, $u_{\ell}$ has a neighbor $u_{r}^{(2)}$ where $(u_{r}^{(2)},u_{\ell}^{(2)})\in M$. By the above argument, there exists an alternating path $P_{2}$ from $u_{\ell}^{(2)}$ to some $M$-free vertex $v_{r}$ of length $O(\log(n)/\gamma)$. So the concatenated path $P=(u_{\ell},u_{r}^{(2)})\circ P_{2}$ is the desired path of length $O(\log(n)/\gamma)$. This completes the proof of the first part of \Cref{lem:short-path-deg-gap}.

\Cref{lem:structure GM}(\ref{enu:random walk}) analogously implies the ``moreover'' part of \Cref{lem:short-path-deg-gap}. The concrete algorithm description is as follows: 
\begin{quote}
Given a vertex $u\in L$, choose a random free edge $(u,v)\notin M$. If $v$ is $M$-free, then terminate. Otherwise, let $u_{2}\in L$ be the unique vertex where $(v,u_{2})\in M$. Repeat the algorithm from $u_{2}$.
\end{quote}
This ``random alternating walk'' algorithm will arrive at an $M$-free vertex in $\Theta(\log(n)/\gamma^{2})$ steps by \Cref{lem:structure GM} and by the correspondence between paths in $G_{M}$ and alternating paths in $G$ as discussed above. By repeating the random alternating walks for $k=O(\log(n)/\gamma)$ times, the failure probability is at most $(1-\Omega(\gamma))^{k}<1/{\rm poly}(n)$. Since each vertex has access to a binary search tree of its incident edges, each random walk step can be implemented in $O(\log n)$ time. Hence, the total running time is $O(\log n)\times k\times\Theta(\log(n)/\gamma^{2})=\Theta(\log^{3}(n)/\gamma^{3})$.
\end{proof}
The rest of this section is for proving \Cref{lem:structure GM}. 

\paragraph{Preliminaries for \Cref{lem:structure GM}.}

Recall that a directed graph $G=(V,E)$ is \emph{Eulerian} if, for each vertex $v\in V$, the in-degree and the out-degree of $v$ in $G$ are the same. Let $\deg_{G}(v)$ denote the out-degree of $v$ (which is the same as in-degree of $v$). For any vertex set $S\subseteq V$, let $\vol_{G}(S)=\sum_{u\in S}\deg_{G}(u)$ denote the total degree in $S$. 
The conductance $\Phi(G)$ of $G$ is defined as
\begin{equation}
\Phi(G)=\min_{\emptyset\neq S\subset V}\frac{|E(S,V\setminus S)|}{\min\{\vol_{G}(S),\vol_{G}(V\setminus S)\}}.\label{eq:conductance}
\end{equation}
Note that $|E(S,V\setminus S)|=|E(V\setminus S,S)|$ in any Eulerian graph. When $\Phi(G)\ge\phi$, we say that $G$ is a $\phi$-expander.

\paragraph{Proof of \Cref{lem:structure GM}.}
Ultimately, \Cref{lem:structure GM} follows from only that $G$ has a $\gamma$-degree-gap at $X$. 

First, observe that every non-sink vertex $u$ has $\deg_{G_{M}}^{{\rm out}}(u)\ge X$ and $\deg_{G_{M}}^{{\rm in}}(u)\le X(1-\gamma)$. For convenient, we will analyze $G_{M}$ via a supergraph $G_{\Eu}$ of $G_{M}$ obtained from $G_{M}$ by adding parallel edges $(t,u)$ to each non-sink vertex $u$ until the in-degree of $u$ is the same as its out-degree. 

\begin{prop}
$G_{\Eu}$ is Eulerian.
\end{prop}

\begin{proof}
By construction, $\deg_{G_{\Eu}}^{{\rm in}}(u)=\deg_{G_{\Eu}}^{{\rm {\rm out}}}(u)$ for all $u\neq t$. It remains to prove that $\deg_{G_{\Eu}}^{{\rm in}}(t)=\deg_{G_{\Eu}}^{{\rm {\rm out}}}(t)$. This follows because, for any directed graph $H$, we have $\sum_{u}\deg_{H}^{{\rm in}}(u)=\sum_{u}\deg_{H}^{{\rm {\rm out}}}(u)$. 
\end{proof}
For any vertex set $S\subseteq V(G_{\Eu})=V(G_{M})$, we write $\vol_{G_{\Eu}}(S)=\sum_{u\in S}\deg_{G_{\Eu}}^{{\rm out}}(u)=\sum_{u\in S}\deg_{G_{\Eu}}^{{\rm in}}(u)$. Let $\vol_{G_{M}}^{{\rm in}}(S)=\sum_{u\in S}\deg_{G_{M}}^{{\rm in}}(u)$ and $\vol_{G_{M}}^{{\rm out}}(S)=\sum_{u\in S}\deg_{G_{M}}^{{\rm out}}(u)$. 
\begin{prop}
\label{prop:conductance G Eu}The conductance of $G_{\Eu}$ is $\Phi(G_{\Eu})\ge\gamma$.
\end{prop}

\begin{proof}
Let $S\subseteq V(G_{\Eu})$ where $\vol_{G_{\Eu}}(S)\le\vol_{G_{\Eu}}(\overline{S})$ and $\overline{S}=V(G_{\Eu})\setminus S$. We need to show that $|E_{G_{\Eu}}(S,\overline{S})|\ge\gamma\vol_{G_{\Eu}}(S).$ 

In the first case, suppose $t\notin S$. Then 
\begin{align}
|E_{G_{M}}(S,\overline{S})| & =\vol_{G_{M}}^{{\rm out}}(S)-|E_{G_{M}}(S,S)|\nonumber \\
 & \ge\vol_{G_{M}}^{{\rm out}}(S)-\vol_{G_{M}}^{{\rm in}}(S)\nonumber \\
 & \ge\gamma\vol_{G_{M}}^{{\rm out}}(S)\nonumber \\
 & =\gamma\vol_{G_{\Eu}}(S)\label{eq:expand}
\end{align}
where the two last lines are because, for each $u\neq t$, we have $\deg_{G_{M}}^{{\rm in}}(u)\le(1-\gamma)\deg_{G_{M}}^{{\rm out}}(u)$ and $\deg_{G_{M}}^{{\rm out}}(u)=\deg_{G_{\Eu}}^{{\rm out}}(u).$ Since $|E_{G_{\Eu}}(S,\overline{S})|\ge|E_{G_{M}}(S,\overline{S})|$, we conclude that $|E_{G_{\Eu}}(S,\overline{S})|\ge\gamma\vol_{G_{\Eu}}(S)$.

In the second case, suppose that $t\in S$. Symmetrically, using the fact that $t\notin\overline{S}$, we have 
\begin{align*}
|E_{G_{M}}(\overline{S},S)| & =\vol_{G_{M}}^{{\rm out}}(\overline{S})-|E_{G_{M}}(\overline{S},\overline{S})|\\
 & \ge\vol_{G_{M}}^{{\rm out}}(\overline{S})-\vol_{G_{M}}^{{\rm in}}(\overline{S})\\
 & \ge\gamma\vol_{G_{M}}^{{\rm out}}(\overline{S})\\
 & =\gamma\vol_{G_{\Eu}}(\overline{S})
\end{align*}
Since $G_{\Eu}$ is Eulerian and $G_{\Eu}\supseteq G_{M}$, we have $|E_{G_{\Eu}}(S,\overline{S})|=|E_{G_{\Eu}}(\overline{S},S)|\ge|E_{G_{M}}(\overline{S},S)|$. Since $\vol_{G_{\Eu}}(\overline{S})\ge\vol_{G_{\Eu}}(S)$, we conclude that $|E_{G_{\Eu}}(S,\overline{S})|\ge\gamma\vol_{G_{\Eu}}(S)$ again.
\end{proof}
Given \Cref{prop:conductance G Eu}, we are ready to prove \Cref{lem:structure GM} below.

\paragraph{Proof of \Cref{lem:structure GM}(\ref{enu:exist path}).}
\begin{proof}
It is well-known that when $\Phi(G_{\Eu})\ge\gamma$, then the diameter of $G_{\Eu}$ is at most $O(\frac{\log(E(G_{\Eu}))}{\gamma})=O(\frac{\log n}{\gamma})$. So, by \Cref{prop:conductance G Eu}, for every non-sink vertex $u$, there exists a path $P_{\Eu}$ from $u$ to $t$ in $G_{\Eu}$ of length $O(\frac{\log n}{\gamma})$. 

Although $G_{\Eu}$ is a supergraph of $G_{M}$, it is obtained from $G_{M}$ only by adding edges going out of the sink $t$. So, if $P_{\Eu}$ ever uses edges in $G_{\Eu}\setminus G_{M}$, then $P_{\Eu}$ already reaches $t$ before using edges in $G_{\Eu}\setminus G_{M}$. So there exists a path $P_{M}$ from $u$ to $t$ in $G_{M}$ of length $O(\frac{\log n}{\gamma})$.
\end{proof}
%

\paragraph{Random Walks in Eulerian Expanders.}
Before proving \Cref{lem:structure GM}(\ref{enu:random walk}),
we review facts related to random walks on Eulerian expanders. 

\begin{thm}
\label{thm:random walk expander}Let $G=(V,E)$ be an Eulerian $\phi$-expander such that, for each vertex $v$, the number of self loops at $v$ is at least $\deg(v)/2$. If we start at a vertex $v$ and perform a random walk in $G$ for $k$ steps, then we end at vertex $t$ with probability at least 
\[
\frac{\deg_{G}(t)}{\vol_{G}(V)}-(1-\phi^{2}/2)^{k}\cdot\sqrt{\frac{\deg_{G}(t)}{\deg_{G}(v)}}.
\]
\end{thm}

In particular, after $k\ge O(\log(n)/\phi^{2})$ steps, any random walk end at $t$ with probability $\approx\frac{\deg_{G}(t)}{\vol_{G}(V)}$. The assumption about the self loops is a mild technical condition to ensure that the distribution of the random walks in $G$ is not periodic and does converge. \Cref{thm:random walk expander} follows from the well-known facts in literature. We give its proof for completeness in \Cref{sec:proof random walk}.

\paragraph{Proof of \Cref{lem:structure GM}(\ref{enu:random walk}).}
\begin{proof}
First, we claim $\deg_{G_{\Eu}}(t)\ge\frac{\gamma}{2}\vol_{G_{\Eu}}(V(G_{\Eu}))$. To see this, consider $S=V(G_{\Eu})\setminus\{t\}$. We have $|E_{G_{M}}(S,t)|\ge\gamma\vol_{G_{M}}^{{\rm out}}(S)$ by \Cref{eq:expand}. The inequality follows because $\deg_{G_{\Eu}}(t)=\deg_{G_{M}}^{{\rm in}}(t)=|E_{G_{M}}(S,t)|$ and $\vol_{G_{\Eu}}(V(G_{\Eu}))=\vol_{G_{M}}^{{\rm out}}(S)+\deg_{G_{M}}^{{\rm out}}(t)\le\vol_{G_{M}}^{{\rm out}}(S)+\vol_{G_{M}}^{{\rm in}}(S)\le2\vol_{G_{M}}^{{\rm out}}(S)$.

In order to apply \Cref{thm:random walk expander}, for each vertex $v$ in $G_{\Eu}$, add $\deg_{G_{\Eu}}(v)$ many self-loops at $v$. Let $G'_{\Eu}$ be the resulting graph. We have $\Phi(G'_{\Eu})\ge\Phi(G_{\Eu})/2\ge\gamma/2$ by \Cref{prop:conductance G Eu}. Let $k=100\log(n)/\gamma^{2}$. By applying \Cref{thm:random walk expander} to $G'_{\Eu}$, we have that, for every vertex $v$, a $k$-step random walk from $v$ in $G'_{\Eu}$ will end at $t$ with probability at least 
\[
\pi(t)-(1-\frac{\gamma^{2}}{8})^{k}\sqrt{n^{2}}\ge\frac{\pi(t)}{2}=\frac{\deg_{G'_{\Eu}}(t)}{2\vol_{G'_{\Eu}}(V(G'_{\Eu}))}\ge\frac{\gamma}{4}.
\]

The same holds in $G_{\Eu}$. This is because a random walk in $G'_{\Eu}$ corresponds to a random walk in $G_{\Eu}$ where, at each step, with probability $1/2$, we stay at the same vertex. We can only get a better bound in the number of steps in $G_{\Eu}$ compared to $G'_{\Eu}$.

Again, note that $G_{\Eu}$ is obtained from $G_{M}$ only by adding edges going out of the sink $t$. If the random walk ever needs to use edges in $G_{\Eu}\setminus G_{M}$, then we actually have reach $t$. So the same bound holds for the random walk holds in $G_{M}$.
\end{proof}

\subsection{Approach 1: Random Augmenting Paths}

In this section, we prove \Cref{lem:perfect matching}(\ref{enu:rand}). The proof is quite simple given \Cref{lem:short-path-deg-gap}.
First of all, we implement $\mathrm{Init}(M_{0})$ straightforwardly in $O(|M_{0}|)$ time. 

To implement $\mathrm{Delete}(u,v)$ where $u\in E(L,R)$, we straightforwardly update $G\gets G-(u,v)$ and $M\gets M-(u,v)$ in $O(1)$ time. Let $L^{<}$ denote the set of left-vertices of $G$ whose degree is less than $X$. Whenever $u\in L^{<}$, then we delete a matched edge $(u,v')\in M$, if it exists, in $O(1)$ time. 

We also explicitly maintain $G':=G\setminus L^{<}$. This takes additional $O(X\log n)$ worst-case update time: whenever $u\in L^{<}$, we explicitly remove $\deg_{G}(u)<X$ edges from $G'$ from the binary search trees of their endpoints. So each vertex in $G'$ has access to a binary search tree of edges incident to it. Therefore, $\mathrm{Delete}(e)$ takes $O(X\log n)$ worst-case update time. 

There are two observations before describing the implementation of $\mathrm{Augment}(u)$. First, $G'$ has a $\gamma$-degree gap at $X$. Second, $M$ is a matching in $G'$. This is because whenever $u\in L^{<}$, then we delete a matched edge $(u,v')\in M$ if it exists. Also, we call $\mathrm{Augment}(u)$ only when $u\notin L^{<}$. So, none of vertices in $L^{<}$ is rematched by $M$. 

Thus, $(G',M,u)$ is a valid input for \Cref{lem:short-path-deg-gap} and, to implement $\mathrm{Augment}(u)$, we simply feed  $(G',M,u)$ to \Cref{lem:short-path-deg-gap}. It will then return an alternating path from $u$ in $G'$ in $O(\log^{3}(n)/\gamma^{3})$ time with high probability. Then, we augment $M$ along this path, which matches $u$ as desired.

Although the algorithm is randomized, we use fresh randomness on each call to \Cref{lem:short-path-deg-gap}. So the correctness does not depend on the previous random choices. Thus, the algorithm works against an adaptive adversary.

\subsection{Approach 2: Decremental Shortest Paths in Residual Graphs}

Here, we prove \Cref{lem:perfect matching}(\ref{enu:det}). Using a standard reduction, the algorithm will maintain a graph $\gres$ with dummy sink $t$ such that for any free vertex $u$, a $ut$ path in $\gres$ corresponds to a $ut$ augmenting path in $G$. Thus, whenever the data structure encounters an $\augment(u)$ operation, it suffices to find a path from $u$ to $t$ in $\gres$. To make this last step efficient, the algorithm will make use of an existing data structure for dynamic shortest paths, which was first developed by Even and Shiloach \cite{EvenS81} and later generalized to directed graphs by King \cite{King99}.

\begin{theorem}[\cite{EvenS81,King99}]
\label{thm:es}
Let $G$ be a dynamic directed graph with positive integer weights and a dedicated sink vertex $t$. Let $\dist(v,t)$ be the shorest $vt$ distance in the current version of $G$. There exists a data structure $\es$ that maintains a shortest path tree $T$ to sink $t$ in $G$\footnote{Formally, $T$ is a tree on the vertex $V$ such that for any $v \in V$, the unique simple $vt$ path in $T$ is a shortest $vt$ path in $G$} where $G$ is subject to the following operations:
\begin{itemize}
\item $\delete(u,v)$ -- deletes edge $(u,v)$ from $G$
\item $\Insert(u,v)$ -- inserts edge $(u,v)$ into $G$, subject to the constraint that the insertion of edge $(u,v)$ does not change $\dist(v,t)$ (and hence does not change any distances to $t$).
\item $\Remove(v)$ -- deletes vertex $v$ and all its incident edges.
\end{itemize}
The \emph{total} update time of the data structure over any sequence of updates is $$O(\textrm{total \# updates} + \sum_{v \in V} \degmax(v) \cdot \distmax(v,t)),$$ where $\degmax(v)$ is the number of edges initially incident to $v$ and $\distmax(v,t)$ is the maximum distance from $v$ to $t$ at any point during the update sequence. (Since distances only increase, if $v$ was never removed, then $\distmax(v,t)$ is simply the final $\dist(v,t)$; otherwise, $\distmax(v,t)$ is $\dist(v,t)$ right before $v$ was removed.)
\end{theorem}

\paragraph{Comparison to Standard ES guarantees:}
The guarantees of the above data structure are somewhat more refined than those presented in the original formulation of ES trees \cite{EvenS81, King99}. But it is easy check that the algorithm presented in those earlier papers achieve our guarantees as well. The key guarantee of those algorithms is that the data structure spends $O(\deg(v))$ time every time $\dist(v,t)$ changes. This immediately implies the desired total runtime of $\sum_{v \in V} \degmax(v) \cdot \distmax(v,t))$. (The original ES trees forbid insertions to guarantee that $\dist(v,t)$ can never decrease, and hence changes at most $\distmax(v,t)$ times; this same guarantee holds in our data structure because we only allow insertions that leave $\dist(v,t)$ unchanged.)

The more refined guarantee also allows for the $\Remove(v)$ operation. We implement $\Remove(v)$ by simply deleting every edge incident to $v$ in a single batch, but the standard analysis would then cause the runtime to blow up because now $\dist(v,t) = \infty$. Note, however, that once $v$ is removed in this way, it will never be touched by the data structure again. 
We thus have that $v$ is touched at most $\distmax(v,t)$ times before the call to $\Remove(v)$ and at most one more time when we delete all edges incident to $v$, so the total time spent on $v$ is indeed $O(\degmax(v) \cdot \distmax(v,t))$.

\subsubsection{Residual Graph}
Instead of working with the graph $G = (L,R,E)$ of Lemma \ref{lem:perfect matching}, we work with a standard residual graph.

\begin{definition}[Residual Graph]
\label{defn:gres}
Given an undirected, unweighted bipartite graph $G = (L,R,E)$, we define a \emph{directed} graph $\gres$ as follows.
\begin{itemize}
\item $V(\gres) = V(G) \cup \{t\}$
\item If $(u,v) \in E(G) \setminus M$ with $u \in L$ $v \in R$, then $E(\gres)$ has a directed edge $(u,v)$
\item $(u,v) \in E(G) \cap M$ with $u \in L$ $v \in R$, then $E(\gres)$ has a directed edge $(v,u)$
\item For every unmatched vertex $v$ in $R(G)$ there is an edge $(v,t)$ in $E(\gres)$
\end{itemize}
\end{definition}

\begin{observation}
\label{obs:gres}
For any $u \in L$ There is a one-to-one correspondence between the following two objects: \textbf{1)} $uv$ alternating path in $G$ that start with an edge not in $M$ and end in an $M$-free vertex $v$, and \textbf{2)} a $ut$ path in $\gres$ whose last non-sink vertex is $v$.
\end{observation}

\subsubsection{Using the ES data structure}

We will apply the ES data structure of Theorem \ref{thm:es} to maintain a shortest path tree $T$ in the graph $\gres$ defined above. To implement $\lpm.\augment(u)$, we will simply find the $ut$ path in $T$. Note that augmenting down the resulting path $P$ will change the matched status of all edges on $P$, so we will need to reverse all their directions in $\gres$. We implement this reversal as follows: we first insert all the reverse edges of $P$ and then delete all the original edges of $P$. By the following observation, these insertions do not change distances and are hence valid ES updates.


\begin{observation}
\label{obs:reverse}
Let $T$ be a shortest path tree to $t$ in a directed graph $G$. For any edge $(u,v) \in T$, $T$ is still a shortest path tree to $v$ in $G \cup (v,u)$.
\end{observation}

\paragraph{Weights:} When we maintain our ES data structure on $\gres$, we will set positive integer weights on the edges into $t$; all other edges will always have weight $1$. These weights are only included to guide the ES data structure and keep it efficient. Note that both Observations \ref{obs:gres} and \ref{obs:reverse} remain true when $\gres$ has positive weights.

We give a brief intuition for why such weights are needed. Consider some $\Delete(u,v)$ operation, where $(u,v)$ was in the matching. Since vertex $v$ is now free, the definition of $\gres$ requires us to add an edge $(v,t)$. But if we added an edge $(v,t)$ of weight $1$, then this would \emph{decrease} $\dist(v,t)$, which our $\es$ data structure cannot handle. Instead, let $\dist_{\gres}(v,t)$ be the $vt$ distance right \emph{before} the deletion of $(v,u)$ from $\gres$. To process $\Delete(u,v)$, our algorithm \emph{first} adds an edge $vt$ of weight $\dist_{\gres}(v,t)$ and \emph{then} deletes $(v,u)$; doing it in this order ensures that the insertion does not decrease the $vt$ distance and is hence a valid $\es$ update.

\subsubsection{Pseudocode for $\lpm$}

\paragraph{Epochs} We saw above that deletions will cause the weights of the $\es$ tree to increase. As a result, the more deletions we have the more $\distmax(v,t)$ will increase, and the less efficient the $\es$ tree will become. In order to maintain efficiency, we introduce a new parameter $\epochsize$ (ep for epoch), and every $\epochsize$ calls to $\LPM.\delete$ we reset all weights in $\gres$ to $1$ and build a new $\es$ data structure from scratch. 

\begin{algorithm}[H]
\caption{Maintaining the LPM}
\SetKwFunction{FInitialize}{LPM.Init($M_0$)}
\SetKwFunction{FResetES}{ResetES}
\SetKwProg{Fn}{Function}{:}{}
\Fn{\FInitialize{}}{
    $M \leftarrow M_0$ \\
    \FResetES{}
}

\SetKwProg{Fn}{Function}{:}{}
\Fn{\FResetES{}}{
    Delete any previous $\es$ data structure \\
    Construct the graph $\gres$ from $G$ and $M$, setting all edge weights to $1$ \\
    Initialize a new ES data structure that maintains a shortest path tree $\tres$ to $t$ in $\gres$ \\
    $\NumDeletions \leftarrow 0$

}
\SetKwFunction{FLMPRemove}{Remove}
\SetKwFunction{FLPMDelete}{LPM.Delete}
\SetKwProg{Fn}{Function}{:}{}
\Fn{\FLPMDelete{$(u,v | u \in L)$}}{
    $\NumDeletions ++$ \\
    $G \leftarrow G - (u,v)$ and $M \leftarrow M - (u,v)$ \\
    \If{$\NumDeletions > \epochsize$}{
        \FResetES{} and break.  \label{line:remove}
        
    }
    \If{$\deg_G(u) < X$}{
        \If{$u$ is $M$-matched to some vertex $x$}{
        delete $(u,x)$ from $M$ \\
        Add an edge $(x,t)$ to $\Gres$ of weight $\dist_{\gres}(x,t)$ \\
            \Comment{does not decrease $\dist_{\gres}(x,t)$} \\

         }
        Execute \FLMPRemove{$u$} in the $\es$ data structure and delete $u$ and all its incident edges from both $G$ and $\gres$ \\
        \Comment{$u$ is no longer in $G$ or $\gres$ throughout the entire execution of $\LPM$}\\
    }
    \Else{
        \If{$(u,v)$ was in $M$ before the deletion}{
            \label{line:add-residual-edge} Add an edge $(v,t)$ to $\Gres$ of weight $\dist_{\gres}(v,t)$ \\
            \Comment{does not decrease $\dist_{\gres}(v,t)$} \\
            Delete the edge $(v,u)$ in $\Gres$
        }
    }
}

\SetKwFunction{FLMPAugment}{LPM.Augment}
\SetKwProg{Fn}{Function}{:}{}
\Fn{\FLMPAugment{$(u | u \in L, \deg_G(u) \geq X$}}{
    Let $P$ be shortest $ut$ path in $\tres$, let $v$ be the last vertex on this path before $t$, and let $P_{uv} = P - (v,t)$ \\
    Modify $M$ by augmenting down the $uv$  path in $G$ corresponding to $P_{uv}$ \\
    Insert all edges of the reverse path $P^{\textrm{rev}}_{uv}$ into $\gres$ \\
    \Comment{ these are valid ES insertions by Observation \ref{obs:reverse}}\\
    Delete all edges in $P{uv}$ from $\Gres$ \\
     \Comment{Combined with the previous line, we have flipped edge directions on $P_{uv}$}
}

\end{algorithm}

\begin{remark}
\label{rem:removed}
Once a vertex $v$ is deleted from the graph in Line \ref{line:remove}, $v$ is no longer considered a part of $G$ or $\Gres$. This implies the following invariant.
\end{remark}

\begin{invariant}
\label{inv:LPM-gap}
Throughout the whole execution of $\LPM$, the graph $G$ has a $\gamma$-degree-gap at $X$.
\end{invariant}


\subsubsection{Analysis of $\LPM$}

We consider the evolution of tree $\Tres$ between two $\ResetES$ operations.

\begin{definition}
Define an \emph{epoch} of the $\LPM$ data structure to be any sequence of updates between two $\ResetES$ operations. 
\end{definition}

\begin{observation}
There are at most $\epochsize$ calls to $\LPM.\delete$ in a single epoch.
\end{observation}

From now on, we will fix a single epoch and analyze the total runtime of $\LPM$ during that epoch.

\begin{definition}
Consider any epoch. Whenever $\LPM.\delete(u,v)$ is called, with $v \in R$, we say that $v$ is \emph{affected}. Define $\Raff$ to contain all affected vertices. 
\end{definition}

\begin{observation}
\label{obs:aff-size}
Throughout an epoch, we have $\card{\Raff} \leq \epochsize$. Moreover, if $M$ is the current matching, then for any free vertex $v$, we have $w(v,t) > 1$ only if $v \in \Raff$.
\end{observation}

\paragraph{Structural Lemma} The key to the analysis is showing that within a single epoch, most vertices in $\Tres$ always have low distance to $t$. We cannot directly apply \Cref{lem:short-path-deg-gap} because edges in $\Tres$ can have positive weight. But by \Cref{obs:aff-size}, the number of edges $(v,t)$ with weight $>1$ is at most $\epochsize$. We are thus able to show that the majority of vertices are unaffected and have $\dist(v,t) = O(\log(n)/\gamma)$, as in  \Cref{lem:short-path-deg-gap}. 

To formalize the above intuition, we use the following structural lemma, which is a generalization of \Cref{lem:short-path-deg-gap}. It shows that in a degree-gap graph, if we mark an arbitrary small set of vertices in $R$, then most vertices in $L$ still have a short augmenting path to an unmarked free vertex.



\begin{lemma}
\label{lem:short-path-extended} Let $G=(L,R,E)$ be a bipartite graph with $\gamma$-degree-gap at $X$ and let $\card{L} = n$. Let $M$ be any left-perfect matching in $G$. Then, given any set of marked right-vertices $\Rmark \subseteq R$, there exists a set $\Lfar \subseteq L$ with $\card{\Lfar} \leq \frac{2}{\gamma} \cdot \card{\Rmark}$, such that for every vertex $u \in L \setminus \Lfar$, there exists an alternating path $P$ from $u$ that starts with an edge $e\notin M$, ends at an \underline{unmarked} $M$-free vertex, and has length at most $4\ln(n)/\gamma = O(\log(n)/\gamma)$. 
\end{lemma}

We defer the proof of \Cref{lem:short-path-extended} to the next section. Here we show how to complete the analysis of $\LPM$ using the lemma as a black box.

\begin{corollary}
\label{cor:short-path}
Consider any epoch. There exists a set of vertices $\Vfar$ with $\card{\Vfar} = O(\epochsize/\gamma)$, such that for every $v \in V \setminus \Vfar$, $\distmax(v,t) = O(\log(n)/\gamma)$.\footnote{Since the entire $\ES$-tree data structure is reset at the end of the epoch, $\distmax$ refers to the maximum distance within a single epoch.}
\end{corollary}

\begin{proof}{Proof of Corollary}
Let $G$ be the graph and $M$ the matching at end of the epoch. Recall that $G$ initially has $\gamma$-degree-gap at $X$ by the statement of Lemma \ref{lem:perfect matching}. Observe that $G$ continues to have $\gamma$-degree-gap at $X$ throughout the execution of $\LPM$: the degree constraint for $R$ holds automatically because the update sequence is decremental, and the degree constraint for $L$ holds because Line \ref{line:remove} of algorithm $\LPM$ removes violating vertices from the graph.


Applying \Cref{lem:short-path-extended} with $\Rmark = \Raff$ yields a set $\Lfar$ such that every $v \in L \setminus \Lfar$ has an alternating path of length $O(\log(n)/\gamma)$ to a free vertex $x \notin \Raff$. By observation \Cref{obs:aff-size}, we have $w(x,t) = 1$, so $\dist(v,t) = O(\log(n)/\gamma)$. Since distances in the ES-tree $\Tres$ are only increasing, and since we are considering the graph at the end of the epoch, we have $\distmax(v,t) = \dist(v,t) = O(\log(n)/\gamma)$.

We now need to consider vertices in $R$ at the end of the epoch. Any $v \in R$ that is free and not in $\Raff$ has $\distmax(v,t) = \dist(v,t) = w(v,t) = 1$. Any vertex $v \in R$ that is matched to a vertex in $L \setminus \Lfar$ can follow a single edge to its mate and then an augmenting path to $t$, so also has $\distmax(v,t) = \dist(v,t) =O(\log(n)/\gamma)$. This leaves at most $\card{\Raff} + \card{\Lfar} = O(\epochsize/\gamma)$ vertices, all of which can be added to $\Vfar$.
\end{proof}

We couple the above corollary with a bound on $\distmax(v,t)$ that applies even to $v \in \Vfar$.

\begin{claim}
\label{claim:aff-length}
Every vertex $v \in G$ has $\dist_{\gres}(v,t) = O((1+\card{\Raff}) \cdot \frac{\log(n)}{\gamma}) = O(\epochsize \cdot \frac{\log(n)}{\gamma})$
\end{claim}

\begin{proof}
Recall that each  $\LPM.\delete$ adss at most one vertex to $\Raff$. We will prove the following subclaim for some large enough constant $C$: 
while $\card{\Raff} \leq k$, the maximum weight in $\gres$ is $\leq (k+1)C \frac{\log(n)}{\gamma}$; \Cref{claim:aff-length} then follows from \Cref{inv:LPM-gap}  and \Cref{lem:short-path-deg-gap} (recall that only the edges into $t$ can have weight $>1$, so every $vt$ path contains at most one edge of weight $>1$). 

The base case is that while no vertices are affected, all weights in $\Gres$ are $1$, so the subclaim follows from \Cref{lem:short-path-deg-gap}. Now assume that the subclaim holds for $\card{\Raff} \leq k-1$ and let $v_k \in R$ be the kth vertex in $\Raff$, which can only occur as a result of some $\LPM.\delete(u,v_k)$. The $\LPM$ data structure adds edge $(v_k,t)$ to $\gres$; the weight of this edge is the distance $\dist(v_k, t)$ \emph{before} the deletion of $(u,v_k)$. We now complete the proof by showing that before the deletion of $(u,v_k)$, $\dist(v_k,t) \leq kC\log(n)/\gamma$: by Lemma \ref{lem:short-path-deg-gap} there exists an augmenting path of length $\leq C\log(n)/\gamma$ in $G$ from $v_k$ to some free vertex $v' \in R$, and by the induction hypothesis there is an edge $(v',t) \in \gres$ of weight $w(v',t) \leq (k-1)C\log(n)/\gamma$.
\end{proof}

\noindent We are now ready to analyze the total update time of the $\ES$ tree.

\begin{lemma}
\label{lem:running-time-es}
The total running time of data structure $\LPM$ during a single epoch is $O(\Delta\cdot n \cdot  \frac{\log(n)}{\gamma} + \Delta\frac{\epochsize^2\log(n)}{\gamma^2})$
\end{lemma}

\begin{proof}
The running time is clearly dominated by the time to maintain the $\es$ data structure. Recall that the running time of the latter is 
$$O(\textrm{total \# updates} + \sum_{v \in V} \degmax(v) \cdot \distmax(v,t)).$$

We bound the second term by

\begin{align*}
\sum_{v \in V} \degmax(v) \cdot \distmax(v,t) &\leq \\
\Delta \sum_{v \in V} \distmax(v,t)   &\leq  
\quad (\Delta \textrm{ is max degree}) 
\\
\Delta \sum_{v \in V \setminus \Vfar}\distmax(v,t) + \Delta \sum_{v \in \Vfar}  \distmax(v,t) 
&\leq 
\\
\Delta \cdot \card{V \setminus \Vfar} \cdot \frac{\log(n)}{\gamma} + \Delta \cdot \card{\Vfar} \cdot \epochsize \cdot \frac{\log(n)}{\gamma} 
&\leq 
\quad (\textrm{Claims \ref{cor:short-path},\ref{claim:aff-length}})
 \\
\Delta\cdot n \cdot  \frac{\log(n)}{\gamma} + \Delta \cdot \frac{\epochsize^2\log(n)}{\gamma^2}&\phantom{\leq}   \quad \ \ (\textrm{\Cref{cor:short-path}}) 
\end{align*}

We now need to bound the total number of updates to the $\es$ tree. Every $\LPM.\Delete$ operation only causes $O(1)$ updates to the $\es$ tree, for a total of $\epochsize$. It remains to analyze $\LPM.\augment(u)$ operations. Each such operation  involves flipping every edge on the augmenting path, which amounts to $O(\dist(u,t))$ operations to the $\es$ data structure, because $\dist(u,t)$ is an upper bound on the number of edges in the $ut$ augmenting path found by $\LPM.\augment(u)$ (it may not be a tight upper bound because $\gres$ can have weights $>1$). 

Let us consider how many $\augment(u)$ operations can occur within a single epoch. When the epoch starts some vertices can be free, and each of these free vertices can undergo an $\augment$ operation. But for any vertex $u$, every $\augment(u)$ operation in the epoch beyond the first must have been preceded by some $\Delete(u,v)$, since that is the only way $u$ could become free again. Thus in total there can be a single $\augment(u)$ for every vertex $u$, as well as $\epochsize$ additional $\augment$ operations that can distributed arbitrarily among the vertices. By Claim \ref{claim:aff-length}, we have:

$$\textrm{[total \# updates] to } \es \leq \sum_{u \in L} \distmax(u,t) + \epochsize \cdot \frac{\epochsize \log(n)}{\gamma}$$
We conclude the claim by noting that the term $\sum_{u \in L} \distmax(u,t)$ is dominated by term $\sum_{v \in V} \degmax(v) \cdot \distmax(v,t)$, which we already bounded in the long equation above. 
\end{proof}

\begin{proof}[Proof of \Cref{lem:perfect matching}(\ref{enu:det})]
It is easy to check that the total update time of $\LPM$ is dominated by the time to maintain the $\es$ structure. Choosing $\epochsize = \sqrt{n\gamma}$ results in a total of at most $1 + u/\epochsize$ epochs (Recall that $u$ is the total number of calls to $\LPM.\Delete$). Applying Lemma \ref{lem:running-time-es}  yields the bound in Lemma \ref{lem:perfect matching}.
\end{proof}

\subsubsection{Proof of Lemma \ref{lem:short-path-extended}}



Recall that we refer to vertices in $\Rmark$ as marked. Define $L_i(v) \subseteq L$ to be the set of all vertices  $u \in L$ satisfying the following property: the shortest alternating path (w.r.t.~$M$)  from $u$, which starts with an unmatched edge $e \in E \setminus M$ incident on $u$ and ends at an $M$-free vertex in $R \setminus \Rmark$, has length exactly $i$. 
Define $R_i(v)$ to be the vertices matched to $L_i(v)$ by $M$: so $|L_i(v)| = |R_i(v)|$. We define $L_{\leq i} = \bigcup_{j \leq i} L_j$ and define $L_{<i}$, $L_{>i}$, and $L_{\geq i}$ analogously. 

The key to the proof is the following claim.

\begin{claim}
\label{claim:short-path-extended-helper}
Say that $\card{L_{> i}} \geq \frac{2}{\gamma}|\Rmark|$. Then, $|L_{i+1}| \geq \frac{\gamma}{2} \cdot |L_{>i}|$. 
\end{claim}

\paragraph{proof of claim:}
We know that $|R_{>i}| = |L_{>i}|$. Define $R' = R_{>i} \cup \Rmark$ and note that by assumption of the claim 
\begin{equation}
\label{eq:Rprime}
|R'| \leq |R_{>i}|(1+\gamma/2) = |L_{>i}|(1+\gamma/2).  
\end{equation} 
Now, clearly 
$$ L_{i+1} = L_{>i} - L_{\geq i+2}$$
We now upper bound $|L_{\geq i+2}|$. It is easy to see that every vertex in $L_{\geq i+2}$ has the property that all of its edges are incident to $R'$. Since $G$ has a $\gamma$-degree gap at $X$, counting the edges implies:
$$X \cdot |L_{\geq i+2}| \leq (1-\gamma)X\cdot |R'|$$
Rearranging and appying \Cref{eq:Rprime} implies
$$|L_{\geq i+2}| \leq (1-\gamma)|R'| \leq (1-\gamma)(1+\gamma/2)|R_{>i}| \leq (1-\gamma/2)|R_{>i}| = (1-\gamma/2)|L_{>i}|$$
This implies that 
$$|L_{i+1}| = |L_{>i}| - |L_{\geq i+2}| \geq \frac{\gamma}{2} \cdot |L_{>i}|$$

\paragraph{Proof of \Cref{lem:short-path-extended} given \Cref{claim:short-path-extended-helper}}
Let $\delta = \frac{4}{\gamma}\ln(n)$. We will show that $\card{L_{> \delta}} \leq \frac{2}{\gamma}\card{\Rmark}$, which clearly implies \Cref{lem:short-path-extended}. Say for contradiction that $\card{L_{>\delta}} > \frac{2}{\gamma}\card{\Rmark}$. This clearly implies that $\card{L_{> j}} > \frac{2}{\gamma}\card{\Rmark}$ for any $j < \delta$. But this means that for any $j < \delta$ we can apply \Cref{claim:short-path-extended-helper}, which in turn implies that for any $j < \delta$ we have 
$$\card{L_{> j}} = \card{L_{>j+1}} + \card{L_{j+1}} \geq \card{L_{>j+1}} + \frac{\gamma}{2} \cdot \card{L_{>j}} \geq (1+\frac{\gamma}{2})\card{L_{>j+1}}$$
The above equation, combined with the fact that for any $x > 1$ we have $(1+1/x)^{2x}>e$,  implies:
$$\card{L} = \card{L_{>0}} \geq \card{L_{> \delta}} \cdot (1+\frac{\gamma}{2})^\delta = \card{L_{> \delta}} \cdot (1+\frac{\gamma}{2})^{\frac{4}{\gamma}\ln(n)}> e^{\ln(n)} = n,$$ which contradicts the definition $\card{L} = n$.


\section{Our Dynamic Maximal Matching Algorithm}
\label{new:dyn:maximal}


In this section, we present our dynamic algorithm for  maximal matching, and prove Theorem~\ref{thm:main}. We fix two parameters $B \in [n]$ and $\epsilon \in (0, 1)$ whose values will be determined later on, and define 
\begin{equation}
\label{eq:delta}
\delta := 100 \epsilon.
\end{equation}
The reader should think of $B = n^{\alpha}$ and $\epsilon = 1/n^{\beta}$, where $0< \beta < \alpha < 1$ are some  absolute constants. 

\subsection{The Framework}
\label{sec:algo:framework}

Let $G = (V, E)$ be the input  graph undergoing edge updates. We maintain a $(B, (1-\epsilon)B)$-{\bf EDCS} $H := (V, E(H))$ of $G$ at all times, as per \cref{thm:edcs}. This incurs an update time of $O\left(\frac{n}{B\epsilon}\right)$. The EDCS $H$ is maintained explicitly, and  we can make adjacency-list and adjacency-matrix queries to $H$ whenever we want. It takes $O(1)$ time to get the answer to any such query. 

Our dynamic maximal matching algorithm  works in {\bf phases}, where each phase consists of $\delta n$ consecutive updates in $G$.  We next introduce the concept of the {\bf core subgraph} $G_{\core} := \left(V, E_{\core}\right)$. At the start of any given phase, we have $E_{\core} := E$. Subsequently, during the phase,  the subgraph $G_{\core}$ evolves as follows. Whenever an edge $e$ gets deleted from $G$, we set $E_{\core} \leftarrow E_{\core} \setminus \{e\}$. In contrast, whenever an edge $e$ gets inserted into $G$, we leave the set $E_{\core}$ unchanged. This ensures that within any given phase $G_{\core}$ is a {\em decremental} graph, with $E_{\core} \subseteq E$ at all times.

Our first task will be to efficiently maintain a  matching $M_{\base} \subseteq E_{\core}$ in the core subgraph, which we will refer to as the {\bf base matching}. This base matching will satisfy certain desirable properties (see \cref{new:lem:base}). Let $G_{\res} := (V_{\res}, E_{\res})$ denote the subgraph of $G$ induced by the set of nodes that are unmatched under $M_{\base}$, that is,  $V_{\res} := V \setminus V(M_{\base})$ and  $E_{\res} := \{ (u, v) \in E : u, v \in V_{\res}\}$. We will refer to $G_{\res}$ as the {\bf adjunct subgraph} of $G$ w.r.t.~$M_{\base}$. Next, we will show how to efficiently maintain a \emph{maximal} matching $M_{\res} \subseteq E_{\res}$ in $G_{\res}$, which will be called the {\bf adjunct matching} (see \cref{new:lem:maximal}). This will immediately imply that $M_{\fnl} := M_{\base} \cup M_{\res}$ is a maximal matching in  $G = (V, E)$; and  $M_{\fnl}$ will be the output of our dynamic algorithm.   

At this point, the reader might get alarmed because  $G_{\res}$ undergoes \emph{node-updates} every time an edge gets added to/removed from $M_{\base}$, which might indicate that it is impossible to efficiently maintain a maximal matching in $G_{\res}$. To overcome this barrier, we will ensure that: (i) the {\em arboricity} of $G_{\res}$ is sublinear in $n$; (ii) the node-updates in $G_{\res}$  satisfy some nice properties, so that coupled with (i) we can maintain the matching $M_{\res}$ in sublinear in $n$ update time. 

\medskip
\noindent {\bf Organization of the rest of \cref{new:dyn:maximal}.} In \cref{sec:classify}, we present a classification of nodes according to their degrees in $H$, which will be used to define the base and the adjunct matchings. In \cref{sec:base,sec:residual}, we respectively show how to maintain the matchings $M_{\base}$ and $M_{\res}$ (see \cref{new:lem:base} and \cref{new:lem:maximal}). Finally, we prove \cref{thm:main} in \cref{sec:wrapup}.

\subsection{Classification of Nodes according to EDCS Degrees}
\label{sec:classify}

To  define the base-matching, we first need to classify the nodes in $V$ into the following categories. Let $G_{\init} = (V, E_{\init})$ and $H_{\init} = (V, E(H_{\init}))$ respectively denote the status of the input graph $G$ and the EDCS $H$ at the start of a given phase. Then throughout that phase:  
\begin{itemize}
\item a node $v \in V$ is {\bf high} if $\deg_{H_{\init}}(v) \geq \left(\frac{1}{2} + \delta - \epsilon \right)B$,
\item a \emph{high} node $v \in V$ is {\bf very-high} if $\deg_{H_{\init}}(v) > \left(\frac{1}{2} + \delta \right)B$,
\item a node $v \in V$ is {\bf medium} if $\left(\frac{1}{2} - \delta \right)B \leq \deg_{H_{\init}}(v) < \left(\frac{1}{2} + \delta - \epsilon \right)B$,
\item a \emph{medium} node $v \in V$ is {\bf almost-low} if $\deg_{H_{\init}}(v) \leq \left( \frac{1}{2} - \delta + \epsilon \right)B$, and 
\item a node $v \in V$ is {\bf low} if $\deg_{H_{\init}}(v) < \left(\frac{1}{2} - \delta \right)B$. 
\end{itemize}
Let $V_{\hi}, V_{\vhi}, V_{\med}, V_{\alo}$ and $V_{\lo}$ respectively denote the sets of high, very-high, medium, almost-low and low nodes. Note that \emph{these sets remain fixed throughout the duration of a phase}; they change only when the current phase ends and the next phase starts. 

\includegraphics[width=\textwidth]{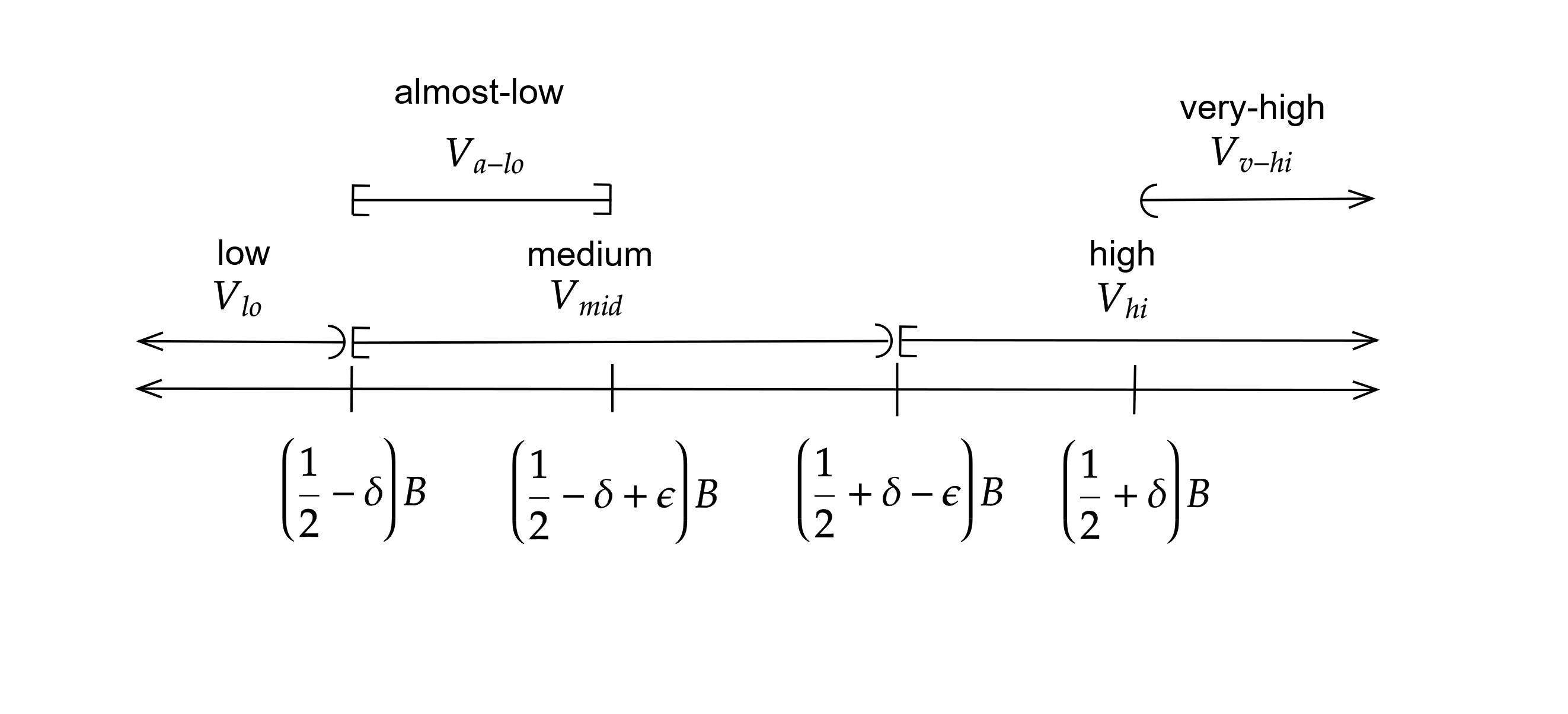}

The main reason why we use this classification is to ensure the proposition below.

\begin{prop}
\label{new:prop:fact}The following properties hold.
\begin{enumerate}
\item \label{new:enu:very-high} Consider any edge $(u, v) \in E(H_{\init})$ where $u$ is very-high. Then $v$ must be low. 
\item \label{new:enu:high}
Consider any edge $(u, v) \in E(H_{\init})$ where  $u$ is high. Then  $v$ must be either low or almost-low.
\item \label{new:enu:low} Consider any edge $e  \in E_{\init}$ whose one endpoint is low and the other endpoint is either low or medium. Then such an edge $e$ must appear in $H_{\init}$.
\end{enumerate}
\end{prop}

\begin{proof} {\bf For part~(\ref{new:enu:very-high})}, note that since $(u, v) \in E(H_{\init})$, we have $\deg_{H_{\init}}(u) + \deg_{H_{\init}}(v) \leq B$ (see \cref{defn:edcs}). Furthermore, since $u$ is very-high, we have $\deg_{H_{\init}}(u) > \left( \frac{1}{2} + \delta \right)B$. Thus, it follows that $\deg_{H_{\init}}(v) \leq B - \deg_{H_{\init}}(u) < \left( \frac{1}{2} - \delta \right)B$, and so $v$ must be a low node.

\medskip
\noindent
{\bf For part~(\ref{new:enu:high})}, note that since $(u, v) \in E(H_{\init})$, we have $\deg_{H_{\init}}(u) + \deg_{H_{\init}}(v) \leq B$ (see \cref{defn:edcs}). Next, since $u$ is high, we have $\deg_{H_{\init}}(u) \geq \left( \frac{1}{2} + \delta - \epsilon \right)B$. Thus, it follows that $\deg_{H_{\init}}(v) \leq B - \deg_{H_{\init}}(u) \leq \left( \frac{1}{2} - \delta + \epsilon \right)B$. Accordingly, the node $v$ is either low or almost-low.

\medskip
\noindent
{\bf For part~(\ref{new:enu:low})}, let $e = (u, v)$, where $u$ is low and $v$ is either low or medium. Then $\deg_{H_{\init}}(u) < \left(\frac{1}{2} - \delta\right)B$ and $\deg_{H_{\init}}(v) < \left(\frac{1}{2} + \delta - \epsilon \right)B$, and hence $\deg_{H_{\init}}(u) + \deg_{H_{\init}}(v) < (1-\epsilon)B$. As $H_{\init}$ is a $(B, (1-\epsilon)B)$-EDCS of $G_{\init}$, it follows that $(u, v) \in E\left(H_{\init}\right)$ (see \cref{defn:edcs}).
\end{proof}

Before concluding this section, we introduce one last category of nodes, as follows. 

\begin{definition}
\label{def:safe}
At any point in time within a phase, let $H_{\core}$ denote the subgraph of $H_{\init}$ restricted to the  edges in $G_{\core}$. Specifically, we define $H_{\core} := \left(V, E(H_{\core})\right)$, where $E(H_{\core}) := E\left(H_{\init}\right) \cap E_{\core}$. Now, consider a high node $v \in V_{\hi}$. At any point in time within the current phase, we say that $v$ is {\bf damaged} if  $\deg_{H_{\core}}(v) < \left( \frac{1}{2} + \delta - 2 \epsilon \right)B$, and {\bf safe} otherwise. We let $V_{\dmg} \subseteq V_{\hi}$ and $V_{\safe} := V_{\hi} \setminus V_{\dmg}$ respectively denote the sets of damaged and safe nodes.
\end{definition}

Recall that   $\deg_{H_{\init}}(v) \geq \left(\frac{1}{2} + \delta - \epsilon \right)B$ for all  nodes $v \in V_{\hi}$. Thus,  for  such a node to get damaged at least $\epsilon B$ many of its incident edges must be deleted since the start of the phase. Since each phase lasts for $\delta n$ updates in $G$ and $\delta = 100 \epsilon$ (see~(\ref{eq:delta})), we get the following important corollary. 

\begin{cor}
\label{cor:damaged}
We always have $\left|V_{\dmg}\right| \leq \frac{2\delta n}{\epsilon B} = O\left(\frac{n}{B} \right)$. Furthermore, at the start of a phase we have $V_{\dmg} = \emptyset$. Subsequently, during the phase, the subset $V_{\dmg}$ grows monotonically over time.
\end{cor}

\subsection{Maintaining the Base Matching}
\label{sec:base}

We will show how to maintain the matching $M_{\base} \subseteq E(H_{\core})$, satisfying two  conditions. \cref{new:cond:base:match} specifies which nodes are matched under $M_{\base}$, whereas \cref{new:cond:base:recourse:1} ensures the ``stability'' of the base matching as the input graph $G$ undergoes updates within a given phase.

\begin{condition}
\label{new:cond:base:match}
Every safe high node is matched under $M_{\base}$, and every damaged high node is free under $M_{\base}$. Furthermore, at most $22\delta n$ medium nodes are free under $M_{\base}$.
\end{condition}

\begin{condition}
\label{new:cond:base:recourse:1}
Consider any update in $G$ within a given phase. Due to this update, at most $O(1)$ nodes get inserted into/deleted from the set $V(M_{\base})$. 
\end{condition}

  Our main result in this section is summarized below.

\begin{lem}
\label{new:lem:base}
We can maintain a  matching $M_{\base} \subseteq E(H_{\core})$  satisfying \cref{new:cond:base:match,new:cond:base:recourse:1} by:
\begin{enumerate}
\item \label{new:alg:det} either a deterministic algorithm with $\tilde{O}\left(B \cdot n^{1/2} \cdot \delta^{-3/2} \right)$ amortized update time; or
\item \label{new:alg:rand} a randomized algorithm with  $\tilde{O}\left(B \cdot \delta^{-1} + \delta^{-3}\right)$ amortized update time. The algorithm is correct with high probability against an adaptive adversary. 
\end{enumerate}
\end{lem}

\medskip
\noindent 
 Our algorithm for maintaining the matching $M_{\base}$ will make use of the following building blocks.
\begin{itemize}
\item A matching $M_{\most} \subseteq E_{\init}$ that will be computed at the start of each phase.
\item A  matching $M_{\hilo}$ in a subgraph $H_{\hilo}$, that will be maintained throughout the phase. 
\end{itemize}
The rest of \cref{sec:base} is organized as follows. In \cref{sub:sec:most}, We first show how to compute the  matching $M_{\most} \subseteq E_{\init}$ at the start of a phase. Next, we define the subgraph $H_{\hilo}$ in \cref{sec:subgraph:hilo}, and summarize some of its key properties. In \cref{sec:init:base} and \cref{sec:update:base}, we respectively show how to initialize the matchings $M_{\hilo}$ and $M_{\base}$ at the start of a phase, and how to maintain $M_{\hilo}$ and $M_{\base}$ throughout the updates within a phase. Finally, we put everything together in \cref{sec:proof:lem:base:new} and prove \cref{new:lem:base}.

 \subsubsection{Computing the matching $M_{\most}$ at the start of a phase}
 \label{sub:sec:most}

Our main result in this section is summarized in the claim below (see Algorithm~\ref{algo:mmost} for the pseudocode).

\begin{claim}
\label{new:cor:mmost} 
\label{new:cor:most:2}
{At the start of a phase}, we can compute a matching $M_{\most} \subseteq E(H_{\init})$, such that at most $6\delta n$ nodes from the set $V_{\hi} \cup V_{\med}$ are free under $M_{\most}$. The matching $M_{\most}$ can be computed 
\begin{itemize}
\item  either by a deterministic algorithm that runs in $\tilde{O}\left(\frac{nB}{\delta}\right)$ time, 
\item  or by a randomized algorithm that runs in $\tilde{O}\left(nB\right)$ time with high probability.
\end{itemize}
\end{claim}

\begin{proof}
{At the start of a phase}, let $H'_{\init}$ be the subgraph of $H_{\init}$ obtained as follows. Initialize $H'_{\init} \leftarrow H_{\init}$. Next, for each very-high node $v \in V_{\vhi}$, arbitrarily keep deleting edges (in $H'_{\init}$)  incident on $v$ until $\deg_{H'_{\init}}(v)$ becomes equal to $\left( \frac{1}{2} + \delta\right)B$. By part~(\ref{new:enu:very-high}) of \cref{new:prop:fact}, this procedure does {\em not} delete any edge incident on a medium or high node. Thus, when this procedure terminates, we have $\left(\frac{1}{2} - \delta \right)B \leq \deg_{H'_{\init}}(v) \leq \left(\frac{1}{2} + \delta \right)B$ for every node $v \in V_{\hi} \cup V_{\med}$, and $\deg_{H'_{\init}}(v) < \left(\frac{1}{2} - \delta \right)B$ for every node $v \in V_{\lo} = V \setminus (V_{\hi} \cup V_{\med})$. Next, we set:
\begin{equation}
\label{eq:set}
\kappa := 3\delta \text{ and } \Delta := \left( \frac{1}{2} + \delta \right)B.
\end{equation}

Since $B = n^{\alpha}$ and $\epsilon = 1/n^{\beta}$, for some absolute constants $0 < \beta < \alpha < 1$ (see the comment  after~(\ref{eq:delta})), and  $\delta = 100 \epsilon$ (see~(\ref{eq:delta})), we infer that $\Delta \geq \frac{4}{\kappa}$. Also, from the preceding discussion, it follows that $\Delta$ is an upper bound on the maximum degree in $H'_{\init}$. We next apply the concerned subroutine (deterministic or randomized) from \cref{lem:match almost max degree}, by setting $G := H'_{\init}$,  and the parameters $\kappa, \Delta$ as in~(\ref{eq:set}). Let $M$ denote the matching returned by the subroutine from Lemma~\ref{lem:match almost max degree}. At this point, we set $M_{\most} \leftarrow M$.

Because of the way we have defined $\kappa, \Delta$ in~(\ref{eq:set}), it follows that  $\left(\frac{1}{2}-\delta\right)B \geq (1-\kappa)\Delta$, and  hence $V_{\kappa} \supseteq  V_{\hi} \cup V_{\med}$ (see \cref{lem:match almost max degree} to recall the definition of $V_{\kappa}$). From \cref{lem:match almost max degree}, we now get:
\begin{itemize}
\item When the phase starts, at most $2\kappa n = 6 \delta n$ nodes from  $V_{\hi} \cup V_{\med}$ are free under $M_{\most}$.
\item Since $H_{\init}$ is a $(B, (1-\epsilon)B)$-EDCS of $G_{\init}$, it has maximum degree at most $B$ (see \cref{defn:edcs}). So, it takes $O(n B)$ time to compute the subgraph $H'_{\init}$ of $H_{\init}$.  After that, we can compute  $M_{\most}$ either in $\tilde{O}(nB/\delta)$ time deterministically, or in $\tilde{O}(nB \delta)$ time with high probability. This leads to an overall deterministic time complexity of $O(nB) + \tilde{O}(nB/\delta) = \tilde{O}(nB/\delta)$, and an overall randomized time complexity of $O(nB) + \tilde{O}(nB\delta) = \tilde{O}(nB)$.
\end{itemize}
This concludes the proof of the claim.
\end{proof}

\begin{algorithm}[H]
\caption{Computing $M_{\most}$. We run  $\texttt{Initialize-}M_{\most}()$ at the start of a phase.}
\label{algo:mmost}

\SetKwFunction{FInitialize}{Initialize-$M_{\most}$}
\SetKwProg{Fn}{Function}{:}{}
\Fn{\FInitialize{}}{
    $H_\init' \leftarrow H_\init$ \\
    \For{$v \in V_{\vhi}$}{
        \While{$\deg_{H_{\init}'}(v) > (\frac{1}{2} + \delta)B$}{
            Let $e$ be an arbitrary edge incident on $v$ in $H'_{\init}$ \\
            $H_{\init}' \leftarrow H_{\init}' \setminus \{e\}$ \\
        }
    }
        $M_{\most} \leftarrow$ apply the algorithm of Lemma~\ref{lem:match almost max degree} on $H_{\init}'$, with $\kappa := 3\delta$ and $\Delta := \left( \frac{1}{2}+\delta\right)B$ \\
}

\end{algorithm}


\subsubsection{The subgraph $H_{\hilo}$}
\label{sec:subgraph:hilo}

We  define the subgraph $H_{\hilo} := (V_{\hilo}, E(H_{\hilo}))$ as follows: 
\begin{equation}
\label{eq:hilo}
E(H_{\hilo}) := \left\{ (u, v) \in E(H_{\core}) : \{u, v\} \cap  V_{\hi} \neq \emptyset \right\} \text{ and } V_{\hilo} := V_{\hi} \cup (V_{\lo} \cup V_{\alo}).
\end{equation}
In words, $E(H_{\hilo})$ is the set of all the edges in $H_{\core}$ incident on the high nodes, and $V_{\hilo}$ is the set of all nodes that are either  hi, or low, or almost-low. Since $E(H_{\core}) \subseteq E(H_{\init})$, part-(\ref{new:enu:high}) of \cref{new:prop:fact} implies that every edge in $E(H_{\hilo})$ has one endpoint in $V_{\hi}$ and the other endpoint in $V_{\lo} \cup V_{\alo}$. Thus, $H_{\hilo}$ is a bipartite subgraph of $H_{\core}$. We will designate the nodes in $V_{\hi}$ as being on the ``left'', and the nodes in $V_{\lo} \cup V_{\alo}$ as being on the ``right''. 

We also define two parameters $X, \gamma$ as in~(\ref{eq:hilo:parameters}), which allow us to derive some useful properties of the subgraph $H_{\hilo}$. This sets up the stage so that we can  apply \cref{lem:perfect matching} on $H_{\hilo}$.

\begin{equation}
\label{eq:hilo:parameters}
X := \left( \frac{1}{2} + \delta - 2\epsilon\right)B \text{ and } \gamma := \frac{\delta}{10^8}.
\end{equation}

\begin{claim}
\label{new:cl:hilo:deg-gap}
Fix $\gamma$ and $X$ as defined in~(\ref{eq:hilo:parameters}). Then $H_{\hilo}$ is a bipartite subgraph of $H_{\core}$ whose edges cross between $V_{\hi}$ (on the left) and $V_{\lo} \cup V_{\alo}$ (on the right). 
\begin{enumerate}
\item \label{enum:hilo:100} At the start of any given phase, $H_{\hilo}$ has an initial $\gamma$-degree gap at $X$. 
\item  Subsequently, during the phase:
\begin{enumerate}
\item \label{enum:hilo:101} We always have $\deg_{H_{\hilo}}(v) \geq X$ for all nodes $v \in V_{\safe} \subseteq V_{\hi}$, and $\deg_{H_{\hilo}}(v) < X$ for all nodes $v \in V_{\dmg} = V_{\hi} \setminus  V_{\safe}$.
\item \label{enum:hilo:102} The set $V_{\safe}$ shrinks monotonically over time ($V_{\safe} = V_{\hi}$ at the start of the phase).
\item \label{enum:hilo:103} An edge insertion in $G$ does not lead to any change in $H_{\hilo}$. In contrast, if an edge $e$ gets deleted from $G$, then the only change it  triggers in $H_{\hilo}$ is this: If $e$ was part of $H_{\hilo}$, then $e$ also gets deleted from $H_{\hilo}$. 
\end{enumerate}
\end{enumerate} 
\end{claim}

\begin{proof}
The bipartiteness of $H_{\hilo}$ follows from the discussion in the paragraph just before~(\ref{eq:hilo:parameters}).

\medskip
\noindent 
{\bf For part-(\ref{enum:hilo:101})}, note that $V_{\safe} \subseteq V_{\hi}$ (see \cref{def:safe}), and for every node $v \in V_{\safe}$ we have $$\deg_{H_{\hilo}}(v) = \deg_{H_{\core}}(v) \geq \left( \frac{1}{2}+\delta - 2 \epsilon \right)B = X.$$
In the above derivation, the first step follows from~(\ref{eq:hilo}), the second step follows from \cref{def:safe}, and the last step follows from~(\ref{eq:hilo:parameters}). Similarly, \cref{def:safe} implies that  $\deg_{H_{\hilo}}(v) = \deg_{H_{\core}}(v) < \left( \frac{1}{2}+\delta - 2 \epsilon \right)B = X$ for all nodes $v \in V_{\dmg} = V_{\hi} \setminus V_{\safe}$.

\medskip
\noindent 
{\bf For part-(\ref{enum:hilo:100})}, consider the situation  at the start of a phase. At this point, we have  $V_{\safe} = V_{\hi}$ (see \cref{def:safe} and \cref{cor:damaged}). Thus, using the same argument as in the preceding paragraph, we conclude that $\deg_{H_{\hilo}}(v) \geq X$ for all nodes $v \in V_{\hi}$. Next, recall the way we  set $X$ and $\gamma$ in~(\ref{eq:hilo:parameters}). Since $\delta = 100  \epsilon$ (see~(\ref{eq:delta})), we get:
\begin{equation}
\label{eq:hilo:parameters:2}
\left( \frac{1}{2} - \delta + \epsilon \right)B \leq (1-\gamma)X.
\end{equation}
Since $E(H_{\hilo}) \subseteq E(H_{\core}) = E(H_{\init})$ at the start of a phase, every node $v \in V_{\lo} \cup V_{\alo}$ satisfies: 
$$\deg_{H_{\hilo}}(v) \leq \deg_{H_{\core}}(v) = \deg_{H_{\init}}(v) \leq  \left( \frac{1}{2} - \delta + \epsilon \right)B \leq (1-\gamma)X.$$
In the above derivation, the penultimate step follows from the definitions of $V_{\lo}$ and $V_{\alo}$, and the last step follows from~(\ref{eq:hilo:parameters:2}). To summarize, at the start of a phase, every node on the left in $H_{\hilo}$ has degree at least $X$, and every node on the right in $H_{\hilo}$ has degree at most $(1-\gamma)X$. Accordingly, $H_{\hilo}$ has an initial $\gamma$-degree gap at $X$.

\medskip
\noindent {\bf Part-(\ref{enum:hilo:102})} follows from \cref{cor:damaged} and \cref{def:safe}.

\medskip
\noindent {\bf Part-(\ref{enum:hilo:103})} follows from the definitions of $H_{\hilo}$ (see~(\ref{eq:hilo})) and $H_{\core}$ (see \cref{def:safe}).
\end{proof}

\subsubsection{Initializing the matchings $M_{\hilo}$ and $M_{\base}$ at the start of a phase}
\label{sec:init:base}

{Throughout this section, we focus on the snapshort of time at the start of a concerned phase}, just before the first update (in that phase) arrives. Thus, for example, {throughout this section we have $V_{\safe} = V_{\hi}$ and $E(H_{\init}) = E(H_{\core})$} (see \cref{def:safe} and \cref{cor:damaged}). Furthermore, recall that $M(S)$ is the set of edges in a matching $M$ that are incident on  $S$ (see \cref{sec:notations}). 

\medskip 
Our main result in this section is summarized below.

\begin{claim}
\label{cl:init:base}
At the start of a phase, we can compute two matchings $M_{\hilo}$ and $M_{\base}$ such that:
\begin{enumerate}
\item \label{enum:init:base:1} $M_{\hilo} = M_{\base}\left(V_{\safe}\right) \subseteq  E(H_{\hilo})$ and $M_{\base} \subseteq E(H_{\core})$.
\item \label{enum:init:base:2} Every node in $V_{\safe}$ is matched in $M_{\base}$, and at most $18\delta n$ nodes in $V_{\med}$ are free in $M_{\base}$.
\item \label{enum:init:base:3} The total time spent to compute $M_{\hilo}$ and $M_{\base}$ is proportional to: $O(nB)$, plus the time it takes to compute $M_{\most}$ (see \cref{new:cor:mmost}), plus the time it takes to make one call to $\mathrm{Init}(.)$ and $6\delta n$ calls to $\mathrm{Augment}(.)$ in $H_{\hilo}$ (see \cref{lem:perfect matching} and part-(\ref{enum:hilo:100}) of \cref{new:cl:hilo:deg-gap}).
\end{enumerate}
\end{claim}

Our algorithm will satisfy the following invariant, which  implies part-(\ref{enum:init:base:1}) of \cref{cl:init:base}.

\begin{invariant}
\label{inv:induction}
$M_{\hilo} = M_{\base}\left( V_{\safe}\right) \subseteq E(H_{\hilo})$, and $M_{\base} \subseteq E(H_{\core})$ is a  valid matching. 
\end{invariant}

\medskip
We devote the rest of this section to the proof of \cref{cl:init:base}. 
Since  $M_{\most}\left(V_{\safe}\right) \subseteq M_{\most} \subseteq E(H_{\init}) = E(H_{\core})$ (see \cref{new:cor:mmost}), from~(\ref{eq:hilo}) it follows that:
\begin{equation}
\label{eq:hilo:init}
M_{\most}\left(V_{\safe}\right) \subseteq E(H_{\hilo}).
\end{equation}

Fix the parameters $\gamma, X$ as in~(\ref{eq:hilo:parameters}), and recall \cref{new:cl:hilo:deg-gap}. {We will run the LPM data structure on $H_{\hilo}$} (see \cref{lem:perfect matching}), {and we will let $M_{\hilo}$ be the matching $M$ maintained by LPM}. 

\begin{wrapper}
We first set $M_{\base} \leftarrow M_{\most}$, and then  call  $\mathrm{Init}(M_{\base}(V_{\safe}))$ in LPM.\footnote{This is a valid operation, since $M_{\most}(V_{\safe}) \subseteq E(H_{\hilo})$ according to~(\ref{eq:hilo:init}).}
\end{wrapper}

\begin{corollary}
\label{cor:calltoinit:new}
At the end of the above call to $\mathrm{Init}(.)$, at most $6\delta n$ nodes from $V_{\safe} \cup V_{\med}$ are free under $M_{\base}$, and \cref{inv:induction} is satisfied.  
\end{corollary}

\begin{proof}
Follows from \cref{new:cor:mmost,new:cl:hilo:deg-gap} and \cref{lem:perfect matching}.
\end{proof}

\begin{wrapper}
Let $F^{\star} = V_{\safe} \setminus V(M_{\hilo})$ denote the set of safe high nodes that are free under $M_{\hilo}$ at this point. By \cref{cor:calltoinit:new} and \cref{new:cl:hilo:deg-gap}, we have $|F^{\star}| \leq 6 \delta n$,  and $\deg_{H_{\hilo}}(v) \geq X$ for all $v \in F^{\star}$. Next, {\bf for each node} $v \in F^{\star}$, we call the subroutine $\mathrm{MatchViaAugment}(v)$, as described below. This concludes our initialization of the matching $M_{\base}$ at the start of a phase.
\end{wrapper}

\begin{wrapper}
\paragraph{The subroutine  $\mathrm{MatchViaAugment}(v)$.}  \ 

\medskip 
We call this subroutine only if $v$ is currently unmatched in $M_{\hilo}$ and $v \in V_{\safe}$, so that $v$ lies on the ``left'' with  $\deg_{\hilo}(v) \geq X$ (see \cref{new:cl:hilo:deg-gap}).   It performs the following operations.

\medskip
\begin{itemize}
\label{alg:matchviaaugment}
\item Make a call to $\mathrm{Augment}(v)$ in LPM (see \cref{lem:perfect matching}). This step matches $v$ under $M_{\hilo}$ by applying an augmenting path,  adds exactly one node (say, $u_v$) besides $v$ to the set $V(M_{\hilo})$, and does not remove any node from the set $V(M_{\hilo})$. Since $v$ lies to the ``left'' of the bipartite graph $H_{\hilo}$, it follows that $u_v$ lies to the ``right'' of $H_{\hilo}$ and hence $u_v \in V_{\lo} \cup V_{\alo}$ (see \cref{new:cl:hilo:deg-gap}). Let $M^{\texttt{old}}_v$ (resp.~$M^{\texttt{new}}_v$)  denote the set of edges that get removed from (resp.~added to) the matching $M_{\hilo}$ as we apply this augmenting path.
\item Set $M_{\base} \leftarrow M_{\base} \cup M^{\texttt{new}}_v \setminus M^{\texttt{old}}_v$.
\item If $u_v$ was previously matched under $M_{\base}$ (say, via an edge $(u_v, u'_v) \in M_{\base}$), then  set $M_{\base} \leftarrow M_{\base} \setminus \{ (u_v, u'_v)\}$. 
\end{itemize}
\end{wrapper}

\begin{claim}
\label{cor:induction} A call to the subroutine $\mathrm{MatchViaAugment}(v)$ satisfies the properties stated below.
\begin{enumerate}
\item \label{enum:induction:1}  It does not violate \cref{inv:induction}.
\item \label{enum:induction:2} It leads to  the node $v$ getting added to the set $V(M_{\base})$, and never leads to a node already in $V(M_{\base}) \cap V_{\safe}$ getting deleted from the set $V(M_{\base})$.
\item \label{enum:induction:3} Overall, it leads to at most two changes (node insertions or deletions) to the set $V(M_{\base})$.
\item \label{enum:induction:4} Its time complexity is proportional to the time it takes to implement the call to $\mathrm{Augment}(v)$.
\end{enumerate}
\end{claim}

\begin{proof} {\bf Part-(\ref{enum:induction:1})} Before the call was made $v$ must have been unmatched by $M_\base$. If $u_v$ was matched by $M_\base$ its previous matching edge is removed in the last step. The remaining edge updates $\mathrm{Augment}(v)$ makes to $M_\base$ consist of removing the edges of $M_\base$ along a path from $v$ to $u_v$ and replacing them with an other matching of the vertices on the same path. Hence, $M_\base$ remains a matching. The augmenting path edges added to $M_\base$ are from $E(H_\hilo)$ by definition hence after the call is made $M_\base(V_\safe) \subseteq E(H_\hilo)$.

\medskip
\noindent 
{\bf For part-(\ref{enum:induction:2})}, the node $v$ gets added to the set $V(M_{\base})$ because of the call to $\mathrm{Augment}(v)$ in LPM (see \cref{lem:perfect matching}), and because we ensure that $M_{\base} \supseteq M_{\hilo}$ (see part-(\ref{enum:induction:1}) above). Furthermore, $u'_v$ is the only node that can potentially get removed from $V(M_{\base})$. Just before the call to $\mathrm{MatchViaAugment}(v)$, we have $(u_v, u'_v) \notin M_{\hilo}$  (otherwise the augmenting path we apply will not end at $u_v$) and $M_{\hilo} = M_{\base}\left( V_{\safe}\right)$ (due to \cref{inv:induction}). Thus, it follows that $u'_v \notin V_{\safe}$.

\medskip
\noindent 
{\bf Part-(\ref{enum:induction:3})} and {\bf part-(\ref{enum:induction:4})} follow from the description of the subroutine.
\end{proof}

\subsubsection*{Proof of \cref{cl:init:base}}

\noindent
{\bf Part-(\ref{enum:init:base:1})} follows from \cref{inv:induction}, \cref{cor:calltoinit:new} and part-(\ref{enum:induction:1}) of \cref{cor:induction}.

\medskip
\noindent 
{\bf For part-(\ref{enum:init:base:2})}, note that $\left| \left(V_{\safe} \cup V_{\med}\right) \setminus V(M_{\base}) \right| \leq 6 \delta n$ at the end of the call to $\mathrm{Init}(.)$ in LPM (see \cref{cor:calltoinit:new}). Subsequently, for each node  $v \in F^{\star} = V_{\safe} \setminus V(M_{\hilo})$, we call $\mathrm{MatchViaAugment}(v)$, which leads to:
\begin{itemize}
\item the node $v$ getting added to $V(M_{\base})$ and  no node  already in $V_{\safe} \cap V(M_{\base})$ getting  deleted from $V(M_{\base})$ (see part-(\ref{enum:induction:2}) of \cref{cor:induction}), and
\item  at most two nodes in $V_{\med}$ getting deleted from $V(M_{\base})$ (see part-(\ref{enum:induction:3}) of \cref{cor:induction}).
\end{itemize}
We make $|F^{\star}| \leq 6 \delta n$ such calls to $\mathrm{MatchViaAugment}(.)$. So, at the end of these calls, all the nodes in $V_{\safe}$ are matched in $M_{\base}$, and at most $6 \delta n+ 2 \cdot |F^{\star}| \leq 18 \delta n$ nodes in $V_{\med}$ are free in $M_{\base}$.

\medskip
\noindent 
{\bf For part-(\ref{enum:init:base:3})}, note that $H_{\hilo}$ is a subgraph of $H_{\core}$ (see \cref{new:cl:hilo:deg-gap}), $H_{\core} = H_{\init}$ at the start of a phase, and $H_{\init}$  is a $(B, (1-\epsilon)B)$-EDCS of $G_{\init}$ (see \cref{sec:classify}). Thus, $H_{\init}$ has maximum degree at most $B$ (see \cref{defn:edcs}), and hence $H_{\hilo}$ also has maximum degree at most $B$. Accordingly, at the start of a phase, we can compute the subgraph $H_{\hilo}$ in $O(n B)$ time. Beyond this $O(nB)$ term, the time taken to initialize $M_{\hilo}$ and $M_{\base}$ is dominated by: the time it takes to compute $M_{\most}$, plus the time complexity of making one call to $\mathrm{Init}(.)$ and $|F^{\star}| \leq 6 \delta n$ calls to $\mathrm{Augment}(.)$ (see \cref{cor:calltoinit:new} and part-(\ref{enum:induction:4}) in \cref{cor:induction}).

\subsubsection{Maintaining the matchings $M_{\hilo}$ and $M_{\base}$ during the updates in a phase}
\label{sec:update:base}

Our main result in this section is summarized in the claim below.

\begin{claim}
\label{cl:maintain:base}
Within a  phase, we can maintain two matchings $M_{\hilo}$ and $M_{\base}$ such that:
\begin{enumerate}
\item \label{enum:maintain:base:1} $M_{\hilo} = M_{\base}\left(V_{\safe}\right) \subseteq  E(H_{\hilo})$ and $M_{\base} \subseteq E(H_{\core})$.
\item \label{enum:maintain:base:2} Every node in $V_{\safe}$ is matched in $M_{\base}$, and at most $22\delta n$ nodes in $V_{\med}$ are free in $M_{\base}$.
\item \label{enum:maintain:base:3} Each update in $G$ leads to at most $O(1)$ changes (node insertions/deletions) in the set $V(M_{\base})$.
\item \label{enum:maintain:base:4} {Excluding the time spent on initializing}  $M_{\base}$ and $M_{\hilo}$ (see \cref{sec:init:base}), the total time spent on maintaining $M_{\base}$ and $M_{\hilo}$ within the phase is proportional to the  time it takes to implement $\delta n$ calls to $\mathrm{Delete}(.)$ and $\mathrm{Augment}(.)$ in $H_{\hilo}$ (see \cref{lem:perfect matching} and \cref{new:cl:hilo:deg-gap}).
\end{enumerate}
\end{claim}

We start by describing how our algorithm handles an update within a phase.

\begin{wrapper}
Both the matchings $M_{\base}$ and $M_{\hilo}$ remain unchanged if an edge gets inserted into the input graph $G = (V, E)$. Henceforth, consider an update that deletes an edge $e$ from $G$. 
\begin{itemize}
\item We first set $M_{\base} \leftarrow M_{\base} \setminus \{e \}$. 
\item By part-(\ref{enum:hilo:103}) of \cref{new:cl:hilo:deg-gap}, this update leads to at most one edge deletion (and zero edge insertions) in $H_{\hilo}$. If no edge gets deleted from $H_{\hilo}$, then the matching $M_{\hilo}$ remains unchanged, and we proceed towards  handling the next update in $G$. 

\item Otherwise,  the same edge $e = (v, u)$,  with $v \in V_{\hi}$ and $u \in V_{\lo} \cup V_{\alo}$ (say), also gets deleted from $H_{\hilo}$ (see \cref{new:cl:hilo:deg-gap}). In this case,  we call $\mathrm{Delete}(v, u)$ in  LPM  (see \cref{lem:perfect matching}). This call to $\mathrm{Delete}(.)$ leads to at most node in $V_{\safe}$ getting unmatched under $M_{\hilo}$, and  if such a node exists, then it must be $v$. Under such an event (i.e., if  $v \in V_{\safe} \setminus V(M_{\hilo})$), we  call the subroutine  $\mathrm{MatchViaAugment}(v)$ as described in \cref{sec:init:base}. When this call terminates, we are ready to handle the next update in $G$.
\end{itemize}
\end{wrapper}

\subsubsection*{Proof of \cref{cl:maintain:base}} \

\noindent 
 {\bf For part-(\ref{enum:maintain:base:1})}, note that \cref{inv:induction} holds before we start processing the very first update in the phase (see~part-(\ref{enum:init:base:1}) of \cref{cl:init:base}). Now, consider any update in $G$ within the phase, and by inductive hypothesis, suppose that \cref{inv:induction} holds just before this update. We have two cases.

\smallskip
\noindent {\em Case (i).} The  update in $G$ does not lead to any change in $H_{\hilo}$. In this case, from~(\ref{eq:hilo}) and \cref{new:cl:hilo:deg-gap}, it is easy to verify that \cref{inv:induction} continues to hold after we finish processing this update.

\smallskip
\noindent {\em Case (ii).} The concerned update in $G$ is the deletion of an edge $e = (v, u)$ that was already part of $H_{\hilo}$. In this case, the edge $e$ gets deleted from $H_{\hilo}$. We set $M_{\base} \leftarrow M_{\base} \setminus \{e \}$, and then call $\mathrm{Delete}(v, u)$ in LPM. From \cref{lem:perfect matching} and \cref{new:cl:hilo:deg-gap}, it is easy to verify that \cref{inv:induction} holds when we finish the call to $\mathrm{Delete}(v, u)$. At this point, if we have $v \in V_{\safe} \setminus V(M_{\hilo})$, then finally we call the subroutine $\mathrm{MatchViaAugment}(v)$. Part-(\ref{enum:induction:1}) of \cref{cor:induction} implies that \cref{inv:induction} continues to hold when we finish the call to $\mathrm{MatchViaAugment}(v)$.

\medskip
\noindent 
{\bf For part-(\ref{enum:maintain:base:2})}, note that before we start processing the very first update in the phase, every node in $V_{\safe}$ is matched in $M_{\base}$, and at most $18 \delta n$ nodes in $V_{\med}$ are free in $M_{\base}$ (see~part-(\ref{enum:init:base:2}) of \cref{cl:init:base}). 

Now, while handling an update  within the phase, the (potential) call to $\mathrm{MatchViaAugment}(v)$ ensures that every node in $V_{\safe}$ continues to remain matched in $M_{\base}$ (see part-(\ref{enum:induction:2}) of \cref{cor:induction}), and  it leads to at most two nodes changing their matched status in $M_{\base}$ (see part-(\ref{enum:induction:3}) of \cref{cor:induction}). Finally, the steps we perform before making the (potential) call to $\mathrm{MatchViaAugment}(v)$ can also lead to at most two nodes changing their matched status in $M_{\base}$. Taken together, this implies that while handling an update within the phase, there are at most four nodes in $V_{\med} \subseteq V$ that change their status from ``matched'' to ``free'' in $M_{\base}$.

Since a phase lasts for $\delta n$ updates, it follows that at any point in time within the phase we have at most $18 \delta n + 4 \cdot \delta n = 22 \delta n$ nodes in $V_{\med
}$ that are free in $M_{\base
}$.

\medskip
\noindent 
{\bf Part-(\ref{enum:maintain:base:3})} follows from part-(\ref{enum:induction:3}) of \cref{cor:induction}.

\medskip
\noindent 
{\bf For part-(\ref{enum:maintain:base:4})}, recall that a phase lasts for $\delta n$ updates in $G$, and the time spent on handling each such update is dominated by the (at most one) call we make to $\mathrm{AugmentViaMatch}(.)$. The bound on the total time complexity during a phase now follows from part-(\ref{enum:induction:4}) of \cref{cor:induction}.

\subsubsection{Proof of \cref{new:lem:base}}
\label{sec:proof:lem:base:new}

From part-(\ref{enum:maintain:base:1}), part-(\ref{enum:maintain:base:2}) and part-(\ref{enum:maintain:base:3}) of \cref{cl:maintain:base}, we infer that the matching $M_{\base}$ satisfies \cref{new:cond:base:match} and \cref{new:cond:base:recourse:1}. It now remains to analyze the amortized update time of the algorithm.

\medskip
Suppose that we use the \emph{deterministic} algorithms from \cref{lem:perfect matching} and \cref{new:cor:mmost}. First, observe that $E(H_{\hilo}) \subseteq E(H_{\core}) \subseteq E(H_{\init})$, and $H_{\init}$ is a $(B, (1-\epsilon)B)$-EDCS of $G_{\init}$ (see \cref{sec:algo:framework} and \cref{sec:classify}). So, the maximum degree  in $H_{\hilo}$, denoted by $\Delta\left(H_{\hilo}\right)$, is at most the maximum degree in $H_{\init}$. This latter quantity, in turn, is at most $B$ (see \cref{defn:edcs}), and hence:
\begin{eqnarray}
\label{eq:maxdeg}
\Delta\left(H_{\hilo}\right) \leq B.
\end{eqnarray}
Recall the values of $\gamma, X$ from~(\ref{eq:hilo:parameters}) and that of $\delta, \epsilon$ from~(\ref{eq:delta}). Now, by part-(\ref{enum:init:base:3}) of \cref{cl:init:base} and part-(\ref{enum:maintain:base:4}) of \cref{cl:maintain:base}, the time spent on maintaining the matching $M_{\base}$ during a phase is at most:
$$O(nB) + \tilde{O}\left( \frac{nB}{\delta} \right) + \tilde{O}\left( \frac{n \Delta\left(H_{\hilo}\right)}{\gamma} \cdot \left(1 + \frac{\delta n}{\sqrt{n\gamma}} \right) \right) = \tilde{O}\left(\frac{nB}{\delta}  \cdot  \frac{\delta n}{\sqrt{n\delta}}  \right) =  \tilde{O}\left( B n^{1/2} \delta^{-3/2}  \cdot \delta n\right).
$$
Since the phase lasts for $\delta n$ updates in $G$, this implies an amortized update time of $\tilde{O}(B n^{1/2} \delta^{-3/2})$.

\medskip
Next, suppose that we use the \emph{randomized} algorithms from \cref{lem:perfect matching} and \cref{new:cor:mmost}. As before, recall the values of $\gamma, X$ from~(\ref{eq:hilo:parameters}) and that of $\delta, \epsilon$ from~(\ref{eq:delta}). Now, by part-(\ref{enum:init:base:3}) of \cref{cl:init:base} and part-(\ref{enum:maintain:base:4}) of \cref{cl:maintain:base}, the time spent on maintaining the matching $M_{\base}$ during a phase is at most:
$$\tilde{O}\left( nB + n + X \cdot \delta n + \gamma^{-3} \cdot \delta n\right) = \tilde{O}\left( nB + B \cdot \delta n + \delta^{-3} \cdot \delta n\right).$$ Since the phase lasts for $\delta n$ updates in $G$, this implies an amortized update time of: 
$$\tilde{O}( B\delta^{-1} + B + \delta^{-3}) = \tilde{O}\left(\frac{B}{\delta} + \frac{1}{\delta^3}\right).$$

\begin{algorithm}[H]
\caption{Pseudocode of the basic functions for maintaining $M_{hilo}$}
\label{algo:mhilo}
    \Comment{Called one at the start of the phase} \\
\SetKwFunction{FInitialize}{Initialize-$M_{\hilo}$}
\SetKwProg{Fn}{Function}{:}{}
\Fn{\FInitialize{}}{
    $V_{\safe} \leftarrow V_{\hi}$ \\
    $M_\hilo^* \leftarrow M_{\most}(V_{\hi}^{sf})$ \\
}

\SetKwFunction{FDeleteLPM}{Delete}
\SetKwFunction{FAugment}{MatchViaAugment}

    \Comment{Called when edge $(v \in V_\hi, u \in V_\lo \cup V_\alo)$ is deleted from $H_\hilo$} \\
\SetKwFunction{FDeleteHHilo}{Delete-$H_\hilo$}
\SetKwProg{Fn}{Function}{:}{}
\Fn{\FDeleteHHilo{$(v,u)$}}{
        $H_\hilo \leftarrow H_\hilo \setminus (v,u)$
        \If{$(v,u) \in M_\hilo$} {
            \FDeleteLPM{$(v,u)$} \Comment{See Lemma~\ref{lem:perfect matching}} \\
            \If{$v \notin V(M_\hilo)$}{
                \FAugment{$(v)$} \Comment{See Algorithm~\ref{alg:matchviaaugment}}
            }
        }
}
\end{algorithm}

\subsection{Maintaining the Adjunct Matching}
\label{sec:residual}

The purpose of this section is to prove the following lemma (see Algorithm~\ref{alg:mres} for the pseudocode).

\begin{algorithm}
\caption{Pseudo-code of the basic functions for maintaining $M_{\res}$}
\label{alg:mres}
    \Comment{Called at the start of each phase} \\
\SetKwFunction{FInitialize}{Initialize} 
\SetKwProg{Fn}{Function}{:}{}
\Fn{\FInitialize{}}{
    $M_\res = \emptyset$
    \For{$(u,v) \in E_\res$}{
        \If{$u,v \notin V(M_\res)$}{
            $M_\res \leftarrow M_\res \cup \{(u,v)\}$ 
        }
    }
}
    \Comment{Called when an edge is inserted into $M_\res$} \\
\SetKwFunction{FInsertMRes}{Insert-$M_\res$}
\SetKwProg{Fn}{Function}{:}{}
\Fn{\FInsertMRes{$(v,u)$}}{
    \For{$w \in V_\dmg$}{
        Remove $u$ and $v$ from $\F_{(w)}$ 
    }
    $M_\res \leftarrow M_\res \cup \{(u,v)\}$ 
}
    \Comment{Called when a vertex is moved from $V_\safe$ to $V_\dmg$} \\
\SetKwFunction{FInsertVhiDmg}{Insert-$V_\dmg$}
\SetKwProg{Fn}{Function}{:}{}
\Fn{\FInsertVhiDmg{$(v)$}}{
    \If{$v \in V(M_\base)$}{
        $u \leftarrow N_{M_\base}(v)$ \\
        $M_\base \leftarrow M_\base \setminus \{(v,u)\}$ \\
    }
    Initialize $\F_{(v)}$, binary search tree of $\{ u \in V_\res \setminus V(M_\res):(u,v) \in E_\res\}$ 
    \If{$\F_{(v)}$ is not empty}{
        $u \leftarrow$ arbitrary vertex in $\F_{(v)}$ \\
        \FInsertMRes{$(u,v)$} 
    }
}
    \Comment{Called when an edge is inserted into $E_\res$} \\
\SetKwFunction{FInsertERes}{Insert-$E_\res$}
\SetKwProg{Fn}{Function}{:}{}
\Fn{\FInsertERes{$(u,v)$}}{
    $E_\res \leftarrow E_\res \cup \{(u,v)\}$
    \If{$u,v \notin V(M_\base \cup M_\res)$}{
        \FInsertMRes{$(u,v)$} \\
    }
}
    \Comment{Called when an edge is deleted from $M_\res$} \\
\SetKwFunction{FDeleteMRes}{Delete-$M_\res$}
\SetKwProg{Fn}{Function}{:}{}
\Fn{\FDeleteMRes{$(u,v)$}}{
    $M_\res \leftarrow M_\res \setminus \{(u,v)\}$ \\
    \For{$w \in \{u,v\}$}{
        \If{$w \in V_\dmg$}{
            \If{$\F_{(w)}$ is not empty}{
                Let $v'$ be an arbitrary vertex in $\F_{(w)}$ \\
                \FInsertMRes{$(v',w)$} \\
            }
        }
        \Else{
            $S \leftarrow \{v \in V_\res | (w,v) \in E_\res \land v \notin V(M_\res)\}$ \\
            \If{$|S| > 0$}{
                $v' \leftarrow$ arbitrary element of $S$ \\
                \FInsertMRes{$(v',w)$} \\
            }
        }
    }
}

    \Comment{Called when an edge is deleted from $E_\res$} \\
\SetKwFunction{FDeleteERes}{Delete-$E_\res$}
\SetKwProg{Fn}{Function}{:}{}
\Fn{\FDeleteERes($u,v$)}{
    $E_\res \leftarrow E_\res \setminus \{(u,v)\}$ \\
    \If{$(u,v) \notin M_\res$}{
        Update $\F_{(v)}$ and $\F_{(u)}$\\
    }
    \Else {
        \FDeleteMRes{$(u,v)$} \\
    }
}
\end{algorithm}

\begin{lem}
\label{new:lem:maximal} Suppose that we  maintain the matching $M_{\base}$ as per \cref{new:lem:base}. Then with an additive overhead of $\tilde{O}\left(\delta n+ \frac{B}{\delta}+ \frac{n}{\delta B}\right)$ amortized update time, we can deterministically and explicitly maintain a \emph{maximal} matching $M_{\res} \subseteq E_{\res}$ in $G_{\res}$. 
\end{lem}

We devote the rest of \cref{sec:residual} to the proof of \cref{new:lem:maximal}. We start by summarizing a few key properties of the adjunct graph $G_{\res}$ in the claim below. 

\begin{claim}
\label{new:cl:residual}
The adjunct graph $G_{\res}$ satisfies the following properties at all times.
\begin{enumerate}
\item \label{enu:res:0} $V_{\safe} \cap V_{\res} = \emptyset$.
\item \label{enu:res:1} $V_{\dmg} \subseteq V_{\res}$ and $\left|V_{\dmg}\right| = O\left(\frac{n}{B}\right)$. Moreover, $V_{\dmg} = \emptyset$ at the start of a phase.
\item \label{enu:res:2} $\left| V_{\med} \cap V_{\res} \right| \leq 22 \delta n$.
\end{enumerate}
\end{claim}

\begin{proof} 
{\bf Part-(\ref{enu:res:0})} All nodes of $V_\safe$ are matched by $M_\base$ (Condition~\ref{new:cond:base:match}) hence no nodes of $V_\safe$ are in $V_\res = V \setminus V(M_\base)$.

{\bf Part-(\ref{enu:res:1})} $V_\dmg$ is initialized as empty at the start of each phase and any node which is moved to $V_\dmg$ from $V_\safe$ has its matching $M_\base$ edge removed hence $V_\dmg \subseteq V_\res$. $|V_\dmg|  = O\left(\frac{n}{B}\right)$ at all times due to Corollary~\ref{cor:damaged}.

{\bf Part-(\ref{enu:res:2})} At initialization by Claim~\ref{cl:init:base} at most $18 n \delta$ vertices of $V_\med$ are unmatched by $M_\base$. As phase consists of $\delta n$ updates and each update may unmatch at most two nodes of $V_\med$ in $M_\base$ we must have that $|V_\med \cap V_\res| = |V_\med \setminus V(M_\base)| \leq 22 \delta n$.

\end{proof}

\begin{corollary}
\label{new:cor:residual:1}
For all nodes $v \in V_{\res} \setminus V_{\dmg}$, we have $\deg_{G_{\res}}(v) = O\left(\delta n + B + \frac{n}{B}\right)$ at all times.
\end{corollary}

\begin{proof}
Consider any node $v \in V_{\res} \setminus V_{\dmg}$. By part-(\ref{enu:res:0}) of \cref{new:cl:residual}, we have $v \notin V_{\safe}$, and hence $v \notin V_{\hi}$ (see \cref{def:safe}). As the node-set $V$ is partitioned into the subsets $V_{\lo}, V_{\med}, V_{\hi}$ (see \cref{sec:classify}), it follows that $v \in V_{\lo} \cup V_{\med}$. 

Since each phase lasts for $\delta n$ updates in $G$, we infer that $|E \setminus E_{\init}| \leq \delta n$ (see \cref{sec:classify}). Thus, the node $v$ can have at most $\delta n$ incident edges in $G_{\res}$ that are part of $E \setminus E_{\init}$. Henceforth, we focus on bounding the number of incident edges on $v$ in $G_{res}$ that are part of $E \cap E_{\init}$. We start by partitioning this concerned set of edges into three subsets, depending on their other endpoints: 
\begin{eqnarray*}
Z_{\lo} & := & \{ (u, v) \in E_{\res} \cap E_{\init} : u \in V_{\lo} \cap V_{\res} \}, \\
Z_{\med} & := & \{ (u, v) \in E_{\res} \cap E_{\init} : u \in V_{\med} \cap V_{\res} \}, \\
Z_{\hi} & := & \{ (u, v) \in E_{\res} \cap E_{\init} : u \in V_{\hi} \cap V_{\res} \}.
\end{eqnarray*}
By part-(\ref{new:enu:low}) of \cref{new:prop:fact}, we have $Z_{\lo} \subseteq E(H_{\init})$. Since $H_{\init}$ is a $(B, (1-\epsilon)B)$-EDCS of $G_{\init}$, it has maximum degree $B$ (see \cref{defn:edcs}), and hence we get $|Z_{\lo}| \leq B$. Next, part-(\ref{enu:res:2}) of \cref{new:cl:residual} implies that $|Z_{\med}| \leq 22 \delta n = O(\delta n)$. Finally, since $V = V_{\safe} \cup V_{\dmg}$ (see \cref{def:safe}), by part-(\ref{enu:res:0}) and part-(\ref{enu:res:1}), we get $|Z_{\hi}| = O\left(\frac{n}{B}\right)$. 

From the preceding discussion, it follows that $\deg_{G_{\res}}(v) = O\left(\delta n + B + \frac{n}{B}\right)$.
\end{proof}

\begin{corollary}
\label{new:cor:residual:2}
At the start of a phase, we have $|E_{\res}| = O(\delta n \cdot \delta n + B n)$.
\end{corollary}

\begin{proof}
At the start of a phase, we have $V_{\res} \cap V_{\hi} = (V_{\res} \cap V_{\safe}) \cup (V_{\res} \cap V_{\dmg}) = \emptyset$, by \cref{def:safe} and part-(\ref{enu:res:0}) and part-(\ref{enu:res:1}) of \cref{new:cl:residual}. Accordingly, at this point in time, the node-set $V_{\res}$ is partitioned into two subsets: $V_{\res} \cap V_{\lo}$ and $V_{\res} \cap V_{\med}$. 

By part-(\ref{enu:res:2}) of \cref{new:cl:residual}, there are at most $O(\delta n \cdot \delta n)$ edges in $E_{\res}$ whose both endpoints are in $V_{\res} \cap V_{\med}$. 

Next, consider any node $u \in V_{\res} \cap V_{\lo}$. Since $V_{\res} \cap V_{\hi} = \emptyset$ at this point in time, every edge in $E_{\res}$ that is incident on $u$ has its other endpoint in $V_{\lo} \cup V_{\med}$ (see \cref{sec:classify}). Accordingly, by part-(\ref{new:enu:low}) of \cref{new:prop:fact}, every edge in $E_{\res} \subseteq E_{\init}$  that is incident on $u$ is part of  $H_{\init}$. Since $H_{\init}$ is a $(B, (1-\epsilon)B)$-EDCS of $G_{\init}$, it has maximum degree $B$ (see \cref{defn:edcs}). Thus, the node $u \in V_{\res} \cup V_{\hi}$ has degree at most $B$ in $G_{\res}$, at this point in time. We conclude that there are at most $O(nB)$ edges in $E_{\res}$ that have some endpoint in $V_{\res} \cap V_{\lo}$.

From the preceding discussion, it follows that $|E_{\res}| = O(\delta n \cdot \delta n + B n)$.
\end{proof}

\subsubsection{Proof of \cref{new:lem:maximal}.}

At the start of a phase, we initialize $M_{\res}$ to be any arbitrary maximal matching in $G_{\res}$. By \cref{new:cor:residual:2}, this takes $O(\delta n \cdot \delta n + B n)$ time. Since the phase lasts for $\delta n$ updates in $G$, this step incurs an amortized update time of:
\begin{equation}
\label{eq:update:time:0}
O\left(\frac{\delta n \cdot \delta n + B n}{\delta n}\right) = O\left(\delta n + \frac{B}{\delta}\right).
\end{equation}

Subsequently, to handle the updates during a phase, we  use the following auxiliary data structures: Each node $v \in V_{\dmg}$ maintains the set $\F_{\res}(v) := \left\{ u \in V_{\res} \setminus V(M_{\res}) : (u, v) \in E_{\res}\right\}$ of its free neighbors (in $G_{\res}$) under the matching $M_{\res}$.\footnote{Maximality of $M_{\res}$ implies that $\F_{\res}(v) = \emptyset$ for all $v \in V_{\dmg} \setminus  V(M_{\res})$.} Whenever a node $v$ moves from $V_{\safe}$ to $V_{\dmg}$, we spend $\tilde{O}(n)$ time to initialize the set $\F_{\res}(v)$ as a balanced search tree. Within a phase, a  node can move from $V_{\safe}$ to $V_{\dmg}$ at most once (see \cref{def:safe} and \cref{cor:damaged}). Thus, by part-(\ref{enu:res:1}) of \cref{new:cl:residual}, these initializations take $\tilde{O}\left( \frac{n}{B} \cdot n \right) = \tilde{O}\left(\frac{n^2}{B} \right)$ total time during a phase. As each phase lasts for $\delta n$ updates in $G$, this incurs an amortized update time of $\tilde{O}\left( \frac{n^2}{B \cdot \delta n}  \right) = \tilde{O}\left( \frac{n}{\delta B}\right)$.

Since $V_{\res} := V \setminus V(M_{\base})$, \cref{new:cond:base:recourse:1} and \cref{new:lem:base} imply that each edge update in $G$ leads to $O(1)$ node-updates in $G_{\res}$. Moreover, while maintaining the maximal matching  $M_{\res}$ in $G_{\res}$, a node-update  can be handled in $\tilde{O}\left(\delta n + B + \frac{n}{B} \right)$ worst-case  time, for the following reason. 

Suppose that we are trying to find a free neighbor  (if it exists) of some node $v \in V_{\res}$. This can be done in $\deg_{G_{\res}}(v) = O\left(\delta n + B + \frac{n}{B} \right)$ time if $v \in V_{\res} \setminus V_{\dmg}$ (see \cref{new:cor:residual:1}), and in $\tilde{O}(1)$ time if $v \in V_{\dmg}$ (using the set $\F_{\res}(v)$). Further, once we decide to match a node $u \in V_{\res}$ to one of its free neighbors, we can accordingly update all the relevant sets $\F_{\res}(v)$ in $\tilde{O}\left( \left| V_{\dmg}\right| \right) = \tilde{O}\left( \frac{n}{B} \right)$ time (see part-(\ref{enu:res:1}) of \cref{new:cl:residual}).

To summarize, the amortized update time we incur, while maintaining the matching $M_{\res}$ throughout the updates within a given phase, is at most:
\begin{equation}
\label{eq:update:time:2}
\tilde{O}\left(\frac{n}{\delta B} + \delta n + B + \frac{n}{B} \right) = \tilde{O}\left(\delta n + B + \frac{n}{\delta B} \right).
\end{equation}

\cref{new:lem:maximal} follows from~(\ref{eq:update:time:0}) and~(\ref{eq:update:time:2}).

\subsection{Wrapping Up: Proof of Theorem~\ref{thm:main}}
\label{sec:wrapup}

\cref{thm:main} follows from \cref{new:lem:base,new:lem:maximal}.

\medskip
To be a bit more specific, first recall that maintaining the EDCS $H$ incurs an update time of $O\left( \frac{n}{\epsilon B}\right) = O\left( \frac{n}{\delta B}\right)$, according to~(\ref{eq:delta}) and the first paragraph of \cref{sec:algo:framework}. Now, for the deterministic algorithm, part-(\ref{new:alg:det}) of \cref{new:lem:base} incurs an amortized update time of $\tilde{O}\left( B \cdot n^{1/2} \cdot \delta^{-3/2}\right)$.  Thus, by \cref{new:lem:maximal}, the overall amortized update time becomes
$$O\left(\frac{n}{\delta B}\right) + \tilde{O}\left( B \cdot n^{1/2} \cdot \delta^{-3/2}\right) + \tilde{O}\left(\delta n + \frac{B}{\delta} + \frac{n}{\delta B}\right)  = \tilde{O}(n^{8/9}), \text{ for } \delta = \frac{1}{n^{1/9}} \text{ and } B = n^{2/9}.$$

Finally, for the randomized algorithm, part-(\ref{new:alg:rand}) of \cref{new:lem:base} incurs an amortized update time of $\tilde{O}(B \delta^{-1} + \delta^{-3})$. Thus, by \cref{new:lem:maximal}, the overall amortized update time becomes
$$O\left(\frac{n}{\delta B}\right) + \tilde{O}\left(\frac{B}{\delta}+ \frac{1}{\delta^{3}}\right) + \tilde{O}\left(\delta n+ \frac{B}{\delta}+ \frac{n}{\delta B}\right) = \tilde{O}\left( n^{3/4} \right), \text{ for } \delta = \frac{1}{n^{1/4}} \text{ and } B = n^{1/2}.$$

\section{Open Problems}

Our novel approach for maintaining a maximal matching suggests several interesting questions.

\paragraph{The Right Answer.}

Fundamentally, how fast can we maintain a maximal matching against an adaptive adversary? Our approach inherently takes $\Omega(\min\{\sqrt{n},\Delta\})$ update time where $\Delta$ is the maximum degree. This lower bounds simply follows from how it handles edge insertions. During a phase insertions are handled greedily. After each update the algorithm iterates over all edges inserted during the phase incident on affected vertices. Assuming that the length of the phase is $\Omega(\sqrt{n})$ this must take at least $\Omega(\min(\sqrt{n}, \Delta))$ time. 

As at the start of the phase the algorithm needs to construct a maximal matching of the graph the initialization of a phase must take $\Omega(n)$ time. Hence, in order to ensure that the amortized cost of the initializations is $O(\sqrt{n})$ over a phase the phase length has to be $\Omega(\sqrt{n})$. 

Is there an algorithm against an adaptive adversary with $o(\min\{\sqrt{n},\Delta\})$ update time? Otherwise, is there a conditional lower bound refuting this hope, or even refuting polylogarithmic update time? The answer in either direction would be very exciting.

\paragraph{Better Decremental Perfect Matching.}

The algorithm from \Cref{lem:perfect matching} is the main bottleneck in our dynamic maximal matching algorithm. In the \emph{static} setting, we can deterministically find a perfect matching in an $m$-edge graph with $\gamma$-degree-gap with $\tilde{O}(m/\gamma)$ time, for example, using Dinic's flow algorithm. Can we deterministically match this running time in the decremental setting? Such an algorithm would immediately imply a deterministic dynamic maximal matching algorithm with $\tilde{O}(n^{4/5})$ update time via our framework.

We can also hope to improve our randomized update time further to $\tilde{O}(n^{2/3})$. Currently, our amortized time for initializing each phase is $\tilde{O}(\frac{n}{B\delta}+\delta n+\frac{B}{\delta}).$ But we can actually improve this to $\tilde{O}(\frac{n}{B\delta}+\delta n+B)$ via warm-starting. We omit this because it alone does not lead to any improvement and will complicate the algorithm further. But, suppose that, in our randomized algorithm, the time spent on $\mathrm{Augment(\cdot)}$ can be improved from $\tilde{O}(1/\gamma^{3})$ to $\tilde{O}(1/\gamma^{2})$. Combining this with warm-starting would immediately improve the final update time from
\[
\tilde{O}(\frac{n}{B\delta}+\delta n+\frac{B}{\delta}+\frac{1}{\delta^{3}})=\tilde{O}(n^{3/4})\text{ by setting }\delta=\frac{1}{n^{1/4}}\text{ and }B=n^{1/2}
\]
to 
\[
\tilde{O}(\frac{n}{B\delta}+\delta n+B+\frac{1}{\delta^{2}})=\tilde{O}(n^{2/3})\text{ by setting }\delta=\frac{1}{n^{1/3}}\text{ and }B=n^{2/3},
\]
where the improvement from $\frac{1}{\delta^{3}}$ to $\frac{1}{\delta^{2}}$ is from the assumption on $\mathrm{Augment(\cdot)}$.

\paragraph{Worst-case Update Time.}

Interestingly, it is open if there exists a dynamic maximal matching algorithm with $o(n)$ worst-case update time, even against an \emph{oblivious} adversary. Previous algorithms against an oblivious adversary inherently guarantee only amortized update time because they must spend a lot of time once the adversary ``deletes the sampled edge''. 

Our algorithm might be a starting point to resolve this open problem. Indeed, our approach exploits amortization in a much more superficial manner. Within each phase, our randomized implementation inherently runs in worst-case update time. The deterministic implementation relies on periodic recomputation of the 'high-low' matching data-structure. However, at the cost of a small blowup in update time these rebuilds can be avoided.

Thus, it remains to spread the work for preprocessing at the beginning of each phase; the task has been easily handled in many other dynamic problems. The idea is to start the pre-processing of the next phase during the operations of the current phase. Once the current phase comes to an end the algorithm switches to the output of the already initialized next phase. This only comes at the cost of a constant blowup in update time.

However, the switching process between the output of the current and next phase also has to be de-amortized. If the goal of the algorithm is to maintain a constant approximate matching this de-amortization process can be achieved with optimal update time trade-offs (see Solomon and Solomon \cite{SolomonS21}). 

Unfortunately, if our goal is to maintain a maximal matching the swtich between outputs becomes more challenging. A partial transformation of one maximal matching to an other one might leave some nodes with unmatched neighbors unmatched. Trivially, to resolve this we have to iterate over the neighborhood of these problematic nodes. However, here we are able to exploit the structure of our maximal matching. Its possible to ensure that none of these problematic nodes happens to be a high node. This in turn means that they may have a small number of potentially unmatched neighbors we need to iterate over to ensure maximallity.

Note that we do not claim that our algorithm can trivially be improved to have worst-case update time. The above discussion ignores multiple situations which can occur due to updates to the current and next maximal matching maintained by the algorithm during the transformation process which appear technically complex to handle.

\paragraph{Dynamic Symmetry Breaking.}

Within the distributed literature symmetry breaking problems have received significant attention (for a detailed list see \cite{BarenboimEPS16}). In the context of distributed computing the process of solving these problems can be described as follows: initially all nodes are assigned the same state symmetrically. In every round nodes have to decide which sate they take in the next round which often has to be different for neighboring nodes. Hence, the process can be described as a task of symmetry breaking.

Interestingly, various problems studied by the symmetry breaking distribution literature have received research attention in the dynamic model. Namely, $(\Delta+1)$-vertex coloring, maximal matching, and maximal independent set are all known to admit efficient distributed algorithms as well as polylogarithmic update time dynamic algorithms \cite{HenzingerP22, solomon2016fully, chechik2019fully} against an oblivious adversary.

While there are known efficient deterministic distributed algorithms for these problems even against an adaptive adversary only maximal matching and $(\Delta+1)$-vertex coloring is known to admit a sub-linear update time dynamic algorithm. This proposes the question weather further symmetry breaking problems admit efficient adaptive or deterministic dynamic implementations and weather there is a deeper connection between distributed and dynamic algorithms in the context of symmetry breaking problems.




\bibliography{ref_abb,ref}

\appendix



\section{Convergence of Random Walks in Eulerian Expanders}

\label{sec:proof random walk}

In this section, we prove \Cref{thm:random walk expander} for completeness. We follow the notations from \cite{reingold2006pseudorandom}.
We say that a matrix $M\in[0,1]^{V\times V}$ is  a \emph{Markov chain with vertex set $V$ }if $\sum_{v\in V}M(v,u)=1$ for every vertex $u$. Given a directed graph $G$, the \emph{random walk matrix} $M_{G}$ of $G$ is the matrix where, for each entry $(v,u)$, 
\[
M_{G}(v,u)=\frac{|E_{G}(u,v)|}{\deg_{G}^{{\rm out}}(u)}.
\]
Observe that $M_{G}$ is a Markov chain.

We will only consider Eulerian graphs where the stationary distribution has a closed form.
\begin{fact}
[Example 1 of \cite{aksoy2016extreme}]For any Eulerian graph $G$, the \emph{stationary distribution} $\pi\in[0,1]^{V}$ of the random walk is such that, for each vertex $v\in V$, $\pi(v)=\deg_{G}(v)/\vol_{G}(V).$
\end{fact}

For any $S\subseteq V$, let $\pi(S)=\sum_{u\in S}\pi(u)$. The conductance of a Markov chain is defined as follows. 
\begin{defn}
[Definition 2.4 of \cite{reingold2006pseudorandom}]Let $M$ be a Markov chain with $n$ vertices and $\pi$ a stationary distribution. The conductance of $M$ with respect to $\pi$ is defined to be 
\[
h_{\pi}(M)=\min_{A:0<\pi(A)\le1/2}\frac{\sum_{u\in A,v\notin A}\pi(u)M(v,u)}{\pi(A)}
\]
\end{defn}

For an Eulerian graph $G$, observe that the conductance of $M_{G}$ defined above coincides with the conductance of $G$ defined in \Cref{eq:conductance}. That is, we have
\begin{equation}
\Phi(G)=h_{\pi}(M_{G}).\label{eq:same conductance}
\end{equation}

We will omit the precise definition of the \emph{spectral expansion} $\lambda_{\pi}(M)$ of a Markov chain $M$. See \cite{reingold2006pseudorandom} for details. We only needs the fact that $\lambda_{\pi}(M)$ dictates the rate of convergence of the distribution of random walks to the stationary distribution, formalized as follows. For any vector $x\in\mathbb{R}^{V}$, let 
\[
\|x\|_{\pi}=\sqrt{\sum_{v\in{\rm supp}(\pi)}\frac{x(v)^{2}}{\pi(v)}.}
\]
\begin{lem}
[Lemma 2.2 of \cite{reingold2006pseudorandom}]\label{lem:convergence}Let $\pi$ be a stationary distribution of Markov chain $M$ on $V$, and let $\alpha$ be any distribution on $V$ such that ${\rm supp}(\alpha)\subseteq{\rm supp}(\pi)$. Then 
\[
\|M^{k}\alpha-\pi\|_{\pi}\le\lambda_{\pi}(M)^{k}\|\alpha-\pi\|_{\pi}.
\]
\end{lem}

That is, the smaller $\lambda_{\pi}(M)$ is the faster the convergence. We can bound $\lambda_{\pi}(M)$ via conductance as follows.
\begin{lem}
[Lemma 2.5 of \cite{reingold2006pseudorandom}]\label{lem:spectral at most cond}Let $M$ be a Markov chain on $n$ vertices such that $M(u,u)\ge1/2$ for every $u$, and let $\pi$ be a stationary distribution of $M$. Then $\lambda_{\pi}(M)\le1-h_{\pi}(M)^{2}/2$.
\end{lem}


With all lemmas above, we are now ready to prove \Cref{thm:random walk expander}.

\paragraph{Proof of \Cref{thm:random walk expander}.}
\begin{proof}
We are given an Eulerian $\phi$-expander $G=(V,E)$ with $n$ vertices where, for each vertex $v$, the number of self loops at $v$ is at least $\deg(v)/2$. Let $p$ denote the probability that a $k$-step random walk in $G$ for $k$ steps starting at vertex $v$ ends at vertex $t$. Our goal is show that $|p-\pi(t)|$ is small. 

Let $\alpha$ be the distribution over $V$ where $\alpha(u)=1$ if $u=v$ and $\alpha(u)=0$ for $u\neq v$. Observe that $M_{G}^{k}\alpha$ is precisely the distribution over $V$ after making $k$ random walk steps in $G$ starting from vertex $v$. So $p$ is precisely the $t$-th entry of $M_{G}^{k}\alpha$. By \Cref{lem:convergence}, we have that 
\[
\|M_{G}^{k}\alpha-\pi\|_{\pi}\le\lambda_{\pi}(M_{G})^{k}\|\alpha-\pi\|_{\pi}.
\]

Let $x\in\mathbb{R}^{V}$ be obtained from $M_{G}^{k}\alpha-\pi$ by restricting to only the $t$-th entry (i.e. $x(u)=0$ for all $u\neq t$). On the one hand, we have
\[
\frac{|p-\pi(t)|}{\sqrt{\pi(t)}}=\|x\|_{\pi}\le\|M_{G}^{k}\alpha-\pi\|_{\pi}
\]
one the other hand, we have

\[
\|\alpha-\pi\|_{\pi}\le\sqrt{1/\pi(v)}.
\]
So we conclude that
\[
|p-\pi(t)|\le(\lambda_{\pi}(M_{G})^{k}\sqrt{\frac{\pi(t)}{\pi(v)}}.
\]

To bound $\lambda_{\pi}(M_{G})$, observe that $M_{G}(u,u)\ge1/2$ for every $u$ because the number of self loops at $u$ is at least $\deg(u)/2$. So, by \Cref{lem:spectral at most cond} and \Cref{eq:same conductance}, we have
\[
\lambda_{\pi}(M_{G})\le1-h_{\pi}(M_{G})^{2}/2=1-\Phi(G)^{2}/2\le1-\phi^{2}/2.
\]
Now, we conclude that
\[
\left|p-\frac{\deg_{G}(t)}{\vol_{G}(V)}\right|\le(1-\phi^{2}/2)^{k}\sqrt{\frac{\deg_{G}(t)}{\deg_{G}(v)}}.
\]
\end{proof}


\section{Proof of \Cref{lem:match almost max degree}}
\label{app:proof:lem:match almost max degree}

We will use the following tools.

\begin{thm}[\cite{assadi2024faster}] \label{thm:fair-matching} For graph $G = (V,E)$ with maximum degree $\Delta$ define $V_\Delta = \{v \in V \mid \deg_G(v) = \Delta\}$. Given both adjacency matrix and list query access to $G$ there exists a randomized algorithm which returns matching $M \in E$ such that, for each $v \in V$, it holds that $\Pr[v \notin M(V)] \leq 2/\Delta$ (not independently) in expected running time $O(n \log n)$.
\end{thm}

A $k$-edge-coloring of graph $G$ is a partitioning of the edges $G$ into $k$ different matchings.

\begin{thm}[\cite{elkin2024deterministic}]\label{thm:determinsitic-edgecoloring} There exists a deterministic algorithm which for graph $G$ with maximum degree $\Delta$ and small constant $\epsilon > 0$ returns a $\Delta(1 + \epsilon)$-edge-coloring of $G$ in time $O(m \log (n)/\epsilon)$.
\end{thm}

Now, we are ready to prove \Cref{lem:match almost max degree}.

\begin{proof}
\textbf{Randomized Algorithm:} The randomized algorithm first generates an extended graph $G' = (V',E')$ such that $V' \supseteq V, E' \subseteq E$ and $\deg_{G'}(v) = \Delta$ for all $v \in V_\kappa$. To obtain such a graph from $G$ it introduces $n \kappa$ dummy vertices $V_D$ and let $V' = V_D \cup V$. Set $E_D = \emptyset$ and for any vertex $v \in V_\kappa$ inserts $\Delta - \deg_G(v) \leq \Delta \kappa$ edges between $v$ and $V_D$ into $E_D$ such that the maximum degree of $V_D$ in $E_D$ is $\Delta$. This can be done in a round-robin fashion deterministically. Finally, define $E' = E \cup E_D$. Note, that generating $G'$ takes $O(|V_D| + |E_D|) = O(n \Delta \kappa)$ time and its easy to see that given adjacency matrix and list query access to $G$ we have adjacency matrix and list query access to $G'$.

Note that when the algorithm of Theorem~\ref{thm:fair-matching} is called on $G'$ in $O(|V'| \log |V'|) = O(n \log n)$ expected time it returns a matching $M$ such that every vertex in $V_\kappa$ is not matched by $M'$ with probability at most $2/\Delta$. Therefore, the expected number of unmatched vertices in $V_\kappa$ under $M$ is at most $2|V_\kappa|/\Delta \leq 2n/\Delta \leq n\kappa/2$. Define a matching $M$ returned by the algorithm to be suitable if it doesn't match at most $n\kappa$ vertices of $V_\kappa$. By Markov's inequality we can argue that the algorithm of Theorem~\ref{thm:fair-matching} returns a suitable matching with constant probability. Furthermore, since its running time is $O(n \log n)$ in expectation it runs in $O(n \log n)$ time with constant probability. Hence, through $O(\log n)$ calls to the algorithm and stopping early if the running time exceeds $O(n \log n)$ we can find a suitable matching $M_S$ with high probability in time $O(n \log^2 n)$. The algorithm finally returns $M_F = M_S \cap E$

To bound the number of unmatched vertices of $V_\kappa$ under $M_F$ observe that there are two ways a vertex in $V_\kappa$ may be unmatched by $M_F$: either it is not matched by $M_S$ or it is matched to a vertex in $V_D$ through an edge not in $E$. Both cases may only hold for $\kappa n$ vertices each since $M_S$ is a suitable matching and $|V_D| \leq n\kappa$. Hence, $M_F$ matched all but $2\kappa n$ vertices of $V_\kappa$.

\textbf{Deterministic Algorithm:} To obtain our deterministic algorithm we find a $\Delta(1+\kappa)$-edge-coloring of $G$ using the algorithm of Elkin and Khuzman \cite{elkin2024deterministic} in $O(m \frac{\log n}{\kappa})$ time. The algorithm partitions the edges of the graph into $\Delta(1+\kappa)$ matchings (or colors). Every vertex in $v \in V_\kappa$ has degree at least $\Delta(1-\kappa)$ in $G$, hence it is matched by at least $\Delta(1-\kappa)/\Delta(1+\kappa) \geq 1 - 2\kappa$ fraction of colors. Therefore, in expectation a uniformly selected color matches at least $|V_\kappa| \cdot (1-2\kappa)$ vertices of $V_\kappa$, that is it doesn't match at most $2|V_\kappa| \kappa \leq 2n\kappa$ vertices. Hence, by the pigeonhole principle there must exist a color which doesn't match less then $2\kappa n$ vertices of $V_\kappa$. We may find this color in $O(m)$ time by simply counting the edges incident on $V_\kappa$ in all colors. 
\end{proof}

\end{document}